\documentclass[11pt,draftcls,onecolumn]{IEEEtran}

\usepackage{epsfig}
\usepackage{graphics,color}
\usepackage{verbatim}
\usepackage{times,amsmath}
\usepackage{multirow}
\usepackage{cases}
\usepackage[tight]{subfigure}
\usepackage{amssymb}

\newtheorem{thm}{Theorem}

\newtheorem{lem}{Lemma}
\newtheorem{rem}{Remark}
\newtheorem{definition}{Definition}
\newtheorem{assumption}{Assumption}

\ifCLASSINFOpdf
\else
\fi

\begin{document}
%
\title{A Statistically Modeling  Method for Performance Limits in Sensor Localization}
%
%
%

\author{Baoqi Huang, Tao Li, Brian D.O. Anderson, Changbin Yu$^*$ 
\thanks{B. Huang and B.D.O. Anderson  are with Research School of Engineering, the Australian National University, Canberra, ACT
        2600, Australia, and National ICT Australia Ltd.
   Email:{\tt\small \{Baoqi.Huang, Brian.Anderson\}@anu.edu.au}.}%
\thanks{T. Li is with  the Key Laboratory of Systems and Control, Institute of Systems Science, Academy of Mathematics
and Systems Science, Chinese Academy of Sciences, Beijing 100190, China
       Email:{\tt\small litao@amss.ac.cn}.}%
\thanks{C. Yu is with Research School of Engineering,
   the Australian National University, Canberra, ACT
        2600, Australia.
       Email:{\tt\small Brad.Yu@anu.edu.au}. Corresponding author.
}%
}

%

\maketitle

\begin{abstract}
In this paper, we study performance limits of sensor localization from a novel perspective. Specifically, we consider
the Cram\'{e}r-Rao Lower Bound (CRLB) in single-hop sensor localization using measurements from received signal
strength (RSS), time of arrival (TOA) and bearing, respectively, but differently from the existing work, we
statistically analyze the trace of the associated CRLB matrix (i.e. as a scalar metric for performance limits of sensor
localization) by assuming anchor locations are random. By the Central Limit Theorems for $U$-statistics, we show that
as the number of the anchors increases, this scalar metric is asymptotically normal in the RSS/bearing case, and
converges to a random variable which is an affine transformation of a chi-square random variable of degree $2$ in the
TOA case. Moreover, we provide formulas quantitatively describing the relationship among the mean and standard
deviation of the scalar metric, the number of the anchors, the parameters of communication channels, the noise
statistics in measurements and the spatial distribution of the anchors. These formulas, though asymptotic in the number
of the anchors, in many cases turn out to be remarkably accurate in predicting performance limits, even if the number
is small. Simulations are carried out to confirm our results.
\end{abstract}


\begin{IEEEkeywords}
Cram\'{e}r-Rao lower bound, received signal strength (RSS), time of arrival (TOA), bearing, $U$-statistics.
\end{IEEEkeywords}

%
\IEEEpeerreviewmaketitle

\section{Introduction}
%
%
%
%
\IEEEPARstart{W}{ireless} sensor networks have a wide range of applications nowadays, including military operations,
medical treatments, environmental sensing, water quality monitoring and many others \cite{Akyildiz02,Chong03}. Location
information plays a vital role in those applications, for it is useful to report the geographic origin of events, to
assist in target tracking, to achieve geographic aware routing, to manage sensor networks, to evaluate their coverage,
and so on. A sensor network generally consists of two types of nodes: anchors and sensors. Anchor locations are known a
priori through GPS or manual configurations, while sensor locations are not known and need to be determined through the
procedures of sensor localization. Up to now, considerable efforts have been invested into the development of
localization algorithms, see e.g. \cite{NN01,SRZ03,SPS03,CAM06,FCMA09,KKM09,KKM10}. Take trilateration, the most basic
localization technique, for example: in a two-dimensional plane, the location of a sensor is estimated from the known
locations of at least three non-collinear anchors and the measured distance, e.g. from received signal strength (RSS)
or time of arrival (TOA), to each anchor; this is also termed as single-hop distance-based sensor localization. As an
extension, in the multi-hop case where not every sensor directly refers to a sufficient number of anchors, iterative
trilateration is proposed by using already localized sensors as pseudo-anchors \cite{SPS03}. Besides, approximate
sensor localization can be realized using mere connectivity data between pairs of neighboring nodes, namely
connectivity-based sensor localization, see e.g. \cite{NN01,SRZ03}.

Single-hop sensor localization can be found in many practical localization scenarios, such as source localization and
target tracking. Moreover, in simultaneous localization and mapping (SLAM) \cite{LD91, DNCDC01}, a mobile robot
equipped with a GPS receiver moves in a two-dimensional environment, measures relative location information to various
objects, and then determines the locations of these objects; herein, the positions where the robot makes measurements
can be abstracted as \emph{anchors}, such that the localization procedure is single-hop. In \cite{SOJ05,PBSJ05}, a
mobile anchor(s) is used to assist in sensor localization by providing relative location measurements to sensors at
multiple positions, which is evidently single-hop. Therefore, it is meaningful to study single-hop sensor localization.

Apart from localization algorithms, the performance limit of sensor localization, namely the lowest achievable error
bound for location estimates, also attracts much attention. On the one hand, it provides a measure of theoretically
optimal performance no matter what sensor localization algorithm is applied; on the other hand, it reflects
fundamentals of sensor localization. Since the Cram\'{e}r-Rao lower bound (CRLB) establishes a lower limit (or bound)
on the variance for any unbiased estimator, it has been widely used in the performance analysis of sensor localization,
see e.g. \cite{PH03,Larsson04,MLV06}.

For single-hop sensor localization in a two-dimensional plane, the CRLB is a $2\times2$ matrix and turns out to be
dependent on multiple factors, including measuring techniques, noise statistics of measurements and sensor-anchor
geometries (i.e. relative node locations, or their coordinates). Since the trace of the CRLB matrix is the minimum
\emph{mean square estimation error (MSE)}, it is often used as a {\em scalar metric} for the performance limit
\cite{PH03,CS04}. Provided that the measuring technique and the noise statistics of measurements are both known, the
scalar metric can be regarded as a function of the sensor-anchor geometry. A valuable problem arises to be seeking to
minimize the scalar metric, equivalent to identifying optimal sensor-anchor geometries for sensor localization, and has
been widely studied \cite{NLG84,DH08,BJ09,BFADP10}. Moreover, since the scalar metric can be infinite, implying that a
localization problem is badly conditioned (e.g. the anchors being or nearly being collinear with the sensor) and
localization algorithms almost fail, we should avoid the situations where the scalar metric takes large values.

The conventional CRLB studies assume a deterministic sensor-anchor geometry which is normally unobtainable for a real
localization problem. Instead, a probability measure for the sensor-anchor geometry might be available. For example,
prior to deploying anchors into a field, we can assume a random and uniform distribution for the anchors' positions.
Hence, we are motivated to study the scalar metric based on a statistical sensor-anchor geometry modeling method. This
method is different from the {\em modified CRLB} in the estimation of nonrandom parameters in the presence of unwanted
(or nuisance) parameters, see e.g. \cite{Gini98,GRM98}; to be specific, the nuisance parameters considered in the
modified CRLB are random in real estimation problems, whereas the anchor positions in our case are fixed and known in
real localization problems but are artificially randomized to obtain a better understanding of localization
performance. To the best of our knowledge, this method has never been considered.

Additionally, the following considerations also motivate this study. Firstly, supposing that every possible
sensor-anchor geometry is equi-probable, the distribution of the scalar metric offers a broad, statistical view on the
performance of single-hop sensor localization, in contrast to one deterministic quantity for a given sensor-anchor
geometry. For instance, if the scalar metric hardly takes large values, there is less reason to worry about the
sensor-anchor geometry; otherwise, one must impose proper control on it. Secondly, the mean of the scalar metric
further establishes a lower limit on the performance of single-hop sensor localization given a fixed number of anchors
with undetermined locations; in the situations where sensor-anchor geometries are unknown, e.g. prior to system
deployment, this performance limit is certainly useful. Lastly, the statistical sensor-anchor geometry modeling method
not only provides insights into single-hop sensor localization and in turn guides us in the design and deployment of
wireless sensor networks, but also as a prototype paves the way for dealing with more complicated scenarios of sensor
localization. For instance, in a mobile environment, as may arise with ad-hoc networks, SLAM, mobile anchors assisting
in sensor localization and so on, it is trivial to concentrate on localization performance in one particular time
instant, whereas it is evidently more attractive to grasp the knowledge about the average localization performance over
a period of time and/or in a wide region. Hopefully, these challenges can be addressed by the statistically modeling
method. In summary, statistical sensor-anchor geometry modeling is a powerful method for investigating the performance
limit of sensor localization.

In this paper, we take into account RSS-based, TOA-based and bearing-only localization, respectively, and show that the
scalar metric in each case is essentially a function of $U$-statistics [see Section \ref{sec:ustatistics} for further
details]. Based on the theory of $U$-statistics, we make the following contributions: (i) it is proved that as the
number of the anchors increases, the scalar metric is asymptotically {\em normal} in the RSS/bearing case and converges
to a random variable which is an affine transformation of a {\em chi-square} random variable of degree $2$ in the TOA
case; (ii) the convergence rate in the RSS/bearing case is shown to be as fast as $O(n^{-\frac{1}{2}})$, where $n$ is the
number of the anchors; (iii) the \emph{asymptotic} formulas for the mean and standard deviation of the scalar metric
are derived in both cases; (iv) last but not the least, these formulas are analyzed to demonstrate some properties of
sensor localization and the conclusions are verified by extensive simulations. Although the derived formulas are
asymptotic in the number of the anchors, in many cases they are remarkably accurate in predicting the performance limit
of sensor localization, even if the number of the anchors is small. For instance, when the number of the anchors is as
small as $6$, the formula for the mean of the scalar metric in the TOA case is capable of providing accurate
predictions.

The remainder of this paper is organized as follows. The next section introduces the problem of single-hop sensor
localization using the RSS measurements and formulates the scalar metric. Section \ref{sec:rss} presents the main
results in the RSS case and extends them in the bearing case. Section \ref{sec:toa} shows the results in the TOA case.
Finally, we conclude this paper and shed light on future work in Section \ref{sec:conclusion}.

\section{Problem Formulation}\label{sec:formulation}

In this section, we first formulate the scalar metric for the performance of single-hop sensor localization using RSS
measurements and then define a random sensor-anchor geometry model. Throughout this paper, we shall use the following
mathematical notations: $(\cdot)^T$ denotes transpose of a matrix or a vector; $Tr(\cdot)$ denotes the trace of a
square matrix; $Pr\{\cdot\}$ denotes the probability of an event; $E_{x}(\cdot)$ and $Std_{x}(\cdot)$ denote the
statistical expectation and standard deviation with respect to the subscripted variable $x$.

\subsection{Single-hop Sensor Localization using RSS Measurements}
In a two-dimensional plane, consider a single sensor (or source, target) located at the origin and $n$ (distance or
angle) measurements made to this sensor at $n$ known locations, as illustrated in Fig.~\ref{fig:network}. Here, the $n$
known locations are abstracted as anchors and are labeled $1,\cdots,n$ with the $i$-th anchor's location denoted by
$\mathbf{a}_i$; the true distance between the sensor and the $i$-th anchor is denoted by $d_i=\|\mathbf{a}_i\|$; the
true angle subtended by $\mathbf{a}_i$ and the positive $x$-axis is denoted by $\alpha_i$.

Take a distance-based localization problem for example: at least $3$ non-collinear anchors are required, namely
$n\geq3$; pair-wise distance measurements $\{\hat{d}_i, i=1,\cdots,n\}$ between the sensor and the $n$ anchors are made
and obey certain measurement models; then, the task is to find an estimate of the true sensor location using the
observable set of distance measurements $\{\hat{d}_i, i=1, \cdots, n\}$.

\begin{figure}
\begin{center}
\centering\includegraphics[scale =1]{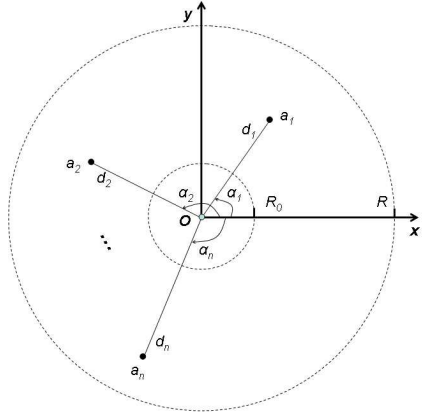} \caption{Single-hop sensor localization.}\label{fig:network}
\end{center}
\end{figure}

Without loss of generality, we let the sensor be a transmitter and the $n$ anchors be receivers. Denote by $\{P_i,
i=1,\cdots,n\}$ the measured received powers at the $n$ anchors transmitted by the sensor, which satisfy the following
assumption.
\begin{assumption}\label{asp:rss2}
The wireless channel satisfies the log-normal (shadowing) model and the received powers $\{P_i, i=1, \cdots, n\}$ at
the $n$ anchors are statistically independent.
\end{assumption}
\begin{rem}
Assumption \ref{asp:rss2} is the basis for converting the RSS measurements (i.e. received powers) to distance estimates
\cite{CD09}, and is  commonly made in both theoretical studies (e.g. \cite{Li07,BJ09,OWL10,SL11}) and experimental
studies (e.g. \cite{PH03,CL09}) on RSS-based sensor localization. It follows that
\begin{equation}
P_i (\text{dBm}) = \overline{P_0}(\text{dBm})-10\alpha\log_{10}\frac{d_i}{R_0}+Z,
\end{equation}
where $P_i(\text{dBm})=10\log_{10} P_i$, $\overline{P_0}(\text{dBm})$ is the mean received power in dBm at a reference
distance $R_0$, $\alpha$ is the path-loss exponent, and $Z$ is a random variable representing the shadowing effect,
normally distributed with mean zero and variance $\sigma_{dB}^2$. As pointed out in \cite{R01}, due to the fact that
the log-normal model does not hold for $d_i=0$, the close-in distance $R_0$ is introduced as the known received power
reference point, and is virtually the lower bound on practical distances used in the wireless communication system.
Further, $\overline{P_0}(\text{dBm})$ is computed from the free space path-loss formula (see, e.g. \cite{R01}).
\end{rem}

\subsection{A Random Sensor-Anchor Geometry Model}
To realize our analysis, we define a random sensor-anchor geometry model by assuming
\begin{assumption}\label{asp:rss1}
The $n$ anchors are randomly and uniformly distributed inside an annulus centered at the sensor and defined by radii
$R_0$ and $R$ ($R>R_0>0$).
\end{assumption}

\begin{rem}
In Assumption \ref{asp:rss1}, $R$ is the upper bound on practical distances which is normally restricted by the factors
determining path-loss attenuations; the lower bound, though set as the close-in distance $R_0$, is mainly devised to
avoid the inconvenience in calculations, and theoretically speaking, can be any arbitrarily small positive number. By
Assumption \ref{asp:rss1}, each possible sensor-anchor geometry is as probable as another, in the sense that the
sensor-anchor geometry follows a ``uniform'' distribution. Besides, $\{d_i, i=1, \cdots, n\}$ and $\{\alpha_i, i=1,
\cdots, n\}$ are mutually independent with probability density functions (pdfs)
\begin{eqnarray}
\label{eqn:pdfdi} f_{d}(x)=\left\{
\begin{array}{ll}
\frac{2x}{R^2-R_0^2},& R_0\leq x\leq R;\\
 0, & \hbox{others}.
 \end{array}\right.
\end{eqnarray}
\begin{eqnarray}
\label{eqn:pdfthetai}
f_{\alpha}(x)=\left\{
\begin{array}{ll}
\frac{1}{2\pi}, & 0\leq x<2\pi;\\
0,& \hbox{others}.
\end{array}\right.
\end{eqnarray}
\end{rem}

\subsection{The Scalar Metric}
Denote by $(\hat x, \hat y)$ the unbiased position estimate of the sensor as well as estimation error (because this
sensor is located at $(0,0)$). Clearly, $\hat x$ and $\hat y$ are dependent on $\mathbf{\upsilon}$ and $\mathbf\omega$,
where
\begin{eqnarray}
 \mathbf{\upsilon}&=& \{P_i,i=1,\cdots,n\}, \\
 \mathbf\omega &=& \{d_i,\alpha_i,i=1,\cdots,n\}.
\end{eqnarray}

The pdf of $P_i$ can be formulated as
\begin{equation}
f_{P_i}(x) = \left\{
\begin{array}{ll}
\frac{10}{(\ln10)\sqrt{2\pi}\sigma_{dB}P_i}\exp\left\{-\frac{b}{2}\left(\ln\frac{d_i}{R_0}+\frac{1}{\alpha}\ln\frac{x}{\overline P_0}\right)^2\right\}, & x>0;\\
0, & x\leq0.
\end{array}\right.
\end{equation}
where $b=\left(\frac{ 10\alpha}{\sigma_{dB}\ln10}\right)^2$.
Then, we can formulate the Fisher information matrix (FIM) $F_{RSS}$ as follows.
\begin{equation}\label{eqn:frss}
F_{RSS}=b\left(
 \begin{array}{cc}
  \sum_{i=1}^n  \frac{\cos^2\alpha_i}{d_i^2} &  \sum_{i=1}^n  \frac{\cos\alpha_i\sin\alpha_i}{d_i^2}\\
  \sum_{i=1}^n  \frac{\cos\alpha_i\sin\alpha_i}{d_i^2}   &  \sum_{i=1}^n  \frac{\sin^2\alpha_i}{d_i^2}
 \end{array}
\right).
\end{equation}
A detailed derivation can be found in \cite{PH03}. If $F_{RSS}$ is non-singular, the CRLB on the covariance of $(\hat
x, \hat y)$, denoted $C_{RSS}$, equals $F_{RSS}^{-1}$ and satisfies
\begin{equation}\label{eqn:lowerbound}
E_{\mathbf{\upsilon}}(\hat x^2+\hat y^2)\geq Tr(C_{RSS}).
\end{equation}
That is, $Tr(C_{RSS})$ is the {\em scalar metric} for the performance limit with the expression (see Appendix
\ref{sec:trace}):
\begin{equation}\label{eqn:trcrlb}
Tr(C_{RSS})=\frac{1}{b}\left( \frac{ \sum_{i=1}^n \frac{1}{d_i^2}}{ \sum_{1\leq i<j\leq n}
\frac{\sin^2(\alpha_i-\alpha_j)}{d_i^2d_j^2}}\right).
\end{equation}

\begin{rem}
Since $Tr(C_{RSS})$ is a function of random variables $\mathbf\omega$, $Tr(C_{RSS})$ itself is a random variable. As
such, by the distribution of $Tr(C_{RSS})$, $E_{\mathbf\omega}(Tr(C_{RSS}))$ and $Std_{\mathbf\omega}(Tr(C_{RSS}))$, we
can statistically investigate the performance limit of single-hop sensor localization over a family of random anchor
locations other than a specific localization problem with given anchor locations. Taking expectations with respect to
$\mathbf\omega$ on both sides of (\ref{eqn:lowerbound}), we can obtain
\begin{equation}
E_{\mathbf\upsilon,\mathbf\omega}(\hat x^2+\hat y^2)\geq E_{\mathbf\omega}(Tr(C_{RSS})),
\end{equation}
namely that $E_{\mathbf\omega}(Tr(C_{RSS}))$ is a lower limit on the performance of localizing one sensor if the
spatial distribution of a fixed number of anchors (with undetermined positions) is known. Considering the involvement
of $E_{\mathbf\omega}(Tr(C_{RSS}))$ and $Std_{\mathbf\omega}(Tr(C_{RSS}))$, we derive a theorem stating their
convergence.
\begin{thm}\label{thm:moments}
Let $Tr(C_{RSS})$ be defined by (\ref{eqn:trcrlb}).
\begin{itemize}
\item If $n\geq5$, then $E_{\mathbf\omega}(Tr(C_{RSS}))<\infty$;
\item If $n\geq9$, then $E_{\mathbf\omega}((Tr(C_{RSS}))^2)<\infty$.
\end{itemize}
\end{thm}
\begin{proof}
See Appendix \ref{sec:moments}.
\end{proof}
\end{rem}

\subsection{$U$-statistics}\label{sec:ustatistics}
$U$-statistics are very natural in statistical work, particularly in the context of independent and identically
distributed (i.i.d.) random variables, or more generally for exchangeable sequences. The origins of the theory of
$U$-statistics are traceable to the seminal paper \cite{Hoeffding1948}, which proved the Central Limit Theorems for
$U$-statistics. Following the publication of this seminal paper, the interest in this class of statistics steadily
increased, crystallizing into a well-defined and vigorously developing line of research in probability theory. The
formal definition for $U$-statistics is presented as follows:
\begin{definition}
Let $\{X_i, i=1,\cdots,n \}$ be i.i.d. $p$-dimensional random vectors. Let $h(x_1,\cdots,x_r)$ be a Borel function on
$\mathbb{R}^{r\times p}$ for a given positive integer $r$ ($\leq n$) and be symmetric in its arguments. A $U$-statistic
$U_n$ is defined by
\begin{equation}
   U_{n}=\frac{r!(n-r)!}{n!} \sum_{1\leq i_1<\cdots<i_r\leq n}h(X_{i_1},\cdots,X_{i_r})
\end{equation}
and $h(x_1,\cdots,x_r)$ is called the kernel of $U_n$.
\end{definition}

It is obvious that $Tr(C_{RSS})$ involves the ratio of two $U$-statistics according to (\ref{eqn:trcrlb}), which
inspires us to study $Tr(C_{RSS})$ through an asymptotic analysis based on the theory of $U$-statistics.

\section{Results with RSS Measurements}\label{sec:rss}
In this section, we endeavor to present an asymptotic analysis of $Tr(C_{RSS})$ based on the theory of $U$-statistics
on account of the difficulty in conducting an accurate and direct analysis.

\subsection{Theories}
First of all, we obtain the following lemma for processing the ratio of two $U$-statistics, which is the basis for
analyzing $Tr(C_{RSS})$ as well as the corresponding scalar metrics in other cases.
\begin{lem}\label{lem:expansion}
Given $\{X^{(1)}_i, i=1, \cdots, n\}$ and $\{X^{(2)}_i, i=1, \cdots, n\}$ where
\begin{itemize}
 \item $\{X^{(1)}_i, i=1, \cdots, n\}$ are i.i.d. random variables with bounded values;
 \item $\{X^{(2)}_i, i=1, \cdots, n\}$ are i.i.d. random variables with bounded values;
 \item $\{X^{(1)}_i, i=1, \cdots, n\}$ and $\{X^{(2)}_i, i=1, \cdots, n\}$ are mutually independent,
\end{itemize}
define two-dimensional vectors $X_i=[\begin{array}{cc}
       X^{(1)}_i & X^{(2)}_i
       \end{array}]^T$
($i=1,\cdots,n$) and two $U$-statistics
\begin{eqnarray}
   T_n &=& \frac{1}{n}\sum_{i=1}^n X^{(1)}_i,  \\
   S_n &=& \frac{2}{n(n-1)} \sum_{1\leq i<j\leq n}\left[X^{(1)}_iX^{(1)}_j\sin^2(X^{(2)}_i-X^{(2)}_j)\right].
\end{eqnarray}

Then
\begin{equation}
\frac{T_n}{S_n}=\frac{1}{m_1m_2}+\frac{2\sigma_1^2}{nm_1^3m_2}+M_n+ R_n\label{eqn:expansion}
\end{equation}
where $m_1=E(X^{(1)}_1)$, $ \sigma_1=Std(X^{(1)}_1)$, $m_2=E(\sin^2(X^{(2)}_1-X^{(2)}_2))$,  $R_n$ is the remainder term, and
\begin{eqnarray}
 M_n &=& \frac{2}{n}\sum_{i=1}^ng_1(X_i)+\frac{2}{n(n-1)}\sum_{1\leq i<j\leq
n}g_2(X_i,X_j),\label{eqn:mn}\\
 g_1(X_i) &=& \frac{m_1-X^{(1)}_i}{2m_1^2m_2}, \\
 g_2(X_i,X_j) &=&
 \frac{1}{m_1m_2}-\frac{X^{(1)}_i+X^{(1)}_j}{m_1^2m_2}+\frac{2X^{(1)}_iX^{(1)}_j}{m_1^3m_2}-\frac{X^{(1)}_iX^{(1)}_j\sin^2(X^{(2)}_i-X^{(2)}_j)}{m_1^3m_2^2}.
\end{eqnarray}
For any $\varepsilon>0$, as $n\to\infty$, $R_n$  satisfies
\begin{eqnarray}
 Pr\left\{|nR_n|\geq \varepsilon\right\} &=& O(n^{-1}), \label{eqn:rn2} \\
 Pr\left\{|n (\ln n) R_n|\geq \varepsilon\right\} &=& o(1). \label{eqn:rn3}
\end{eqnarray}
\end{lem}
\begin{proof}
See Appendix \ref{sec:proofofexpansion}.
\end{proof}

\begin{rem}
By Lemma \ref{lem:expansion}, the ratio of two $U$-statistics (i.e. $\frac{T_n}{S_n}$) can be expanded into a linear
expression consisting of one constant, one term converging to $0$ with the rate $O(n^{-1})$, one $U$-statistic (i.e.
$M_n$) with mean $0$ and one remainder term (i.e. $R_n$) converging to $0$ in probability.
\end{rem}

Obviously, $Tr(C_{RSS})$ can be expanded based on Lemma \ref{lem:expansion}. By letting $X_i^{(1)}=\frac{1}{d_i^2}$ and
$X_i^{(2)}=\alpha_i$,  we can derive $m_2=0.5$, and
\begin{eqnarray}
  m_1 &=& 2\left( \frac{\ln\frac{R}{R_0}}{R^2-R_0^2}\right), \label{eqn:m1}\\
  \sigma_1 &=& \sqrt{\frac{1}{R_0^2R^2}-\left(\frac{2\ln\frac{R}{R_0}}{R^2-R_0^2}\right)^2}. \label{eqn:sigma1}
\end{eqnarray}
One of our main results is further summarized as follows.
\begin{thm}\label{thm:rssconvergence}
Use the same notations $b$ and $Tr(C_{RSS})$ as in Section \ref{sec:formulation} and $m_1$ and $\sigma_1$ as defined by
(\ref{eqn:m1}) and (\ref{eqn:sigma1}). Define a sequence of random variables
\begin{equation}
   W_n=
   \left(\frac{\sqrt{n}(n-1)bm_1^2}{4\sigma_1}\right)Tr(C_{RSS})-\frac{\sqrt{n}m_1}{\sigma_1}-\frac{2\sigma_1}{\sqrt{n}m_1}.
\end{equation}
Then, as $n\to\infty$, $W_n$ converges in distribution to a standard normal random variable.
\end{thm}

\begin{proof}
See Appendix \ref{sec:proofofrssconvergence}.
\end{proof}

\begin{rem}\label{rem:distribution}
In view of the affine relationship between $W_n$ and $Tr(C_{RSS})$, it is straightforward that $Tr(C_{RSS})$ is
asymptotically normal. Therefore, given a sufficiently large $n$, the distribution of $Tr(C_{RSS})$ can be approximated by the
following normal distribution
\begin{equation}\label{eqn:approximate_distribution}
    \mathcal{N}\left(\frac{4}{(n-1)bm_1}+\frac{8\sigma_1^2}{n(n-1)bm_1^3},\left(\frac{4\sigma_1}{\sqrt{n}(n-1)bm_1^2}\right)^2\right).
\end{equation}
Most importantly, the above normal random variable enables us to analytically study the performance limit, i.e.
$Tr(C_{RSS})$. Firstly, we can obtain a comprehensive knowledge about how $Tr(C_{RSS})$ is statistically distributed
and how $Tr(C_{RSS})$ is dependent on $n$. Secondly, using the normal distribution function from
(\ref{eqn:approximate_distribution}), we can compute the probability that $Tr(C_{RSS})$ is below a given threshold for
a known value of $n$; in turn, we can determine a threshold such that $Tr(C_{RSS})$ is below the threshold with a
certain confidence level, say $0.99$; in addition, we can find the minimum $n$ such that $Tr(C_{RSS})$ is below a given
threshold with a certain confidence level. Such analysis is undoubtedly helpful for the design and deployment of sensor
networks. Lastly, the moments of $Tr(C_{RSS})$ can be approximated by the corresponding moments of the normal variable
defined by (\ref{eqn:approximate_distribution}), namely,
\begin{eqnarray}
 E_{\mathbf{\omega}}(Tr(C_{RSS}))  &\approx&  \frac{4}{(n-1)bm_1}+\frac{8\sigma_1^2}{n(n-1)bm_1^3}, \label{eqn:etrrss3}\\
Std_{\mathbf{\omega}}(Tr(C_{RSS})) &\approx&  \frac{4\sigma_1}{\sqrt{n}(n-1)bm_1^2},\label{eqn:stdtrrss2}
\end{eqnarray}
which characterize the relationship among the mean and standard deviation of $Tr(C_{RSS})$, the number of the anchors,
the noise statistics of the RSS measurements and the spatial distributions of the anchors.
\end{rem}

A natural question arises as to how large $n$ should be to obtain a good approximation; this gives rise to the
convergence rate study. In the literature of $U$-statistics, the Berry-Esseen bound was developed for characterizing
the convergence rates of $U$-statistics \cite{GS73,CJV80}. Considering the fact that $W_n$ is affine to a $U$-statistic
(i.e. $M_n$) according to the proof of Theorem $\ref{thm:rssconvergence}$, we propose the following theorem describing
the convergence rate of $W_n$ in the way similar to the Berry-Esseen bound.

\begin{thm}\label{thm:rssrate}
Use the notations in Theorem \ref{thm:rssconvergence} and define
\begin{equation}
   \nu_3=E\left(\left(\frac{1}{d_1^2}-m_1\right)^3\right).
\end{equation}
Then, as $n\to\infty$,
\begin{equation}\label{eqn:beb}
   \sup_x\left|F_n(x)-\Phi(x)\right|\leq
   \left|\left(\frac{\nu_3+\frac{2\sigma_1^4}{m_1}}{6\sigma_1^{3}}\right)\frac{(x^2-1)e^{-\frac{x^2}{2}}}{\sqrt{2\pi}}\right|n^{-\frac{1}{2}}+O(n^{-1})
\end{equation}
where $F_n(x)$ is the distribution function of $W_n$ and $\Phi(x)$ is the standard normal distribution function.
\end{thm}
\begin{proof}
See Appendix \ref{sec:proofofrssrate}.
\end{proof}

\begin{rem}
Theorem \ref{thm:rssrate} shows that as $n\to\infty$, the density of $W_n$ converges to standard normality with the
rate $O(n^{-\frac{1}{2}})$. Additionally, it can be verified that the coefficient associated with $n^{-\frac{1}{2}}$ is
a function of the ratio $\frac{R}{R_0}$; that is to say, the convergence rate of the density of $W_n$ is not determined
by the individual values of $R_0$ and $R$, but by the ratio $\frac{R}{R_0}$.
\end{rem}

\subsection{Simulations}
In this subsection, we would like to carry out simulations to verify Theorem \ref{thm:rssconvergence} and the accuracy
of using (\ref{eqn:etrrss3}) and (\ref{eqn:stdtrrss2}). In the simulations, the parameters describing the wireless
channel, i.e. $\alpha$, $\sigma_{dB}$ and $R_0$, are set as $2.3$, $3.92$ and $1$ m, respectively, which were measured
in a real environment \cite{PH03}.

\begin{figure}
\centering

\subfigure[]{
\includegraphics[scale=0.8]{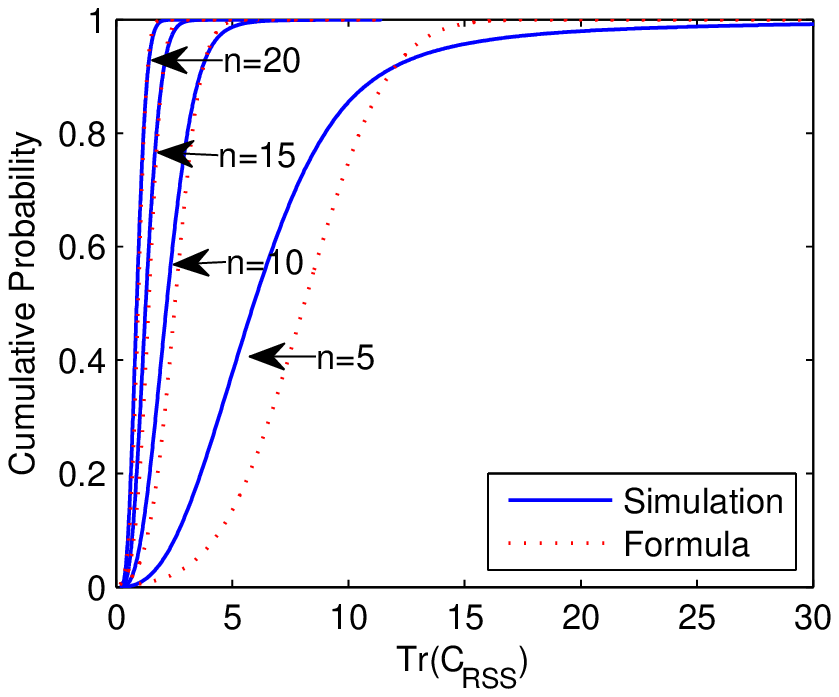}
\label{fig:rss_cdf}} \subfigure[]{
\includegraphics[scale=0.8]{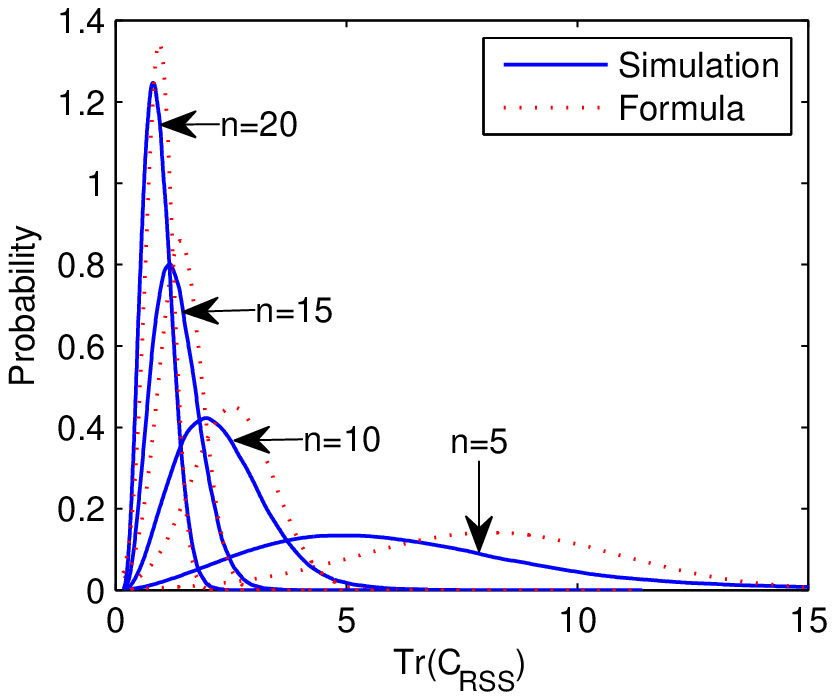}
\label{fig:rss_pdf}} \caption[]{The distribution functions and pdfs of $Tr(C_{RSS})$ with $\alpha=2.3,
\sigma_{dB}=3.92, R_0=1$ m and $R=10$ m.} \label{fig:distribution}
\end{figure}

In the first place, we plot the actual distribution functions of $Tr(C_{RSS})$ from simulations (with the legend
``Simulation'') and the normal distribution functions defined by (\ref{eqn:approximate_distribution}) (with the legend
``Formula'') for $n=5, 10, 15, 20$ in Fig.~\ref{fig:rss_cdf}. As can be seen, when $n=5$, the discrepancy between the
pair of distribution functions is quite obvious, but when $n=10$, the discrepancy becomes very small, and when $n=15$
or $20$, the discrepancy can be negligible. The gradually diminishing discrepancy arises for two reasons: the intrinsic
error in approximating a $U$-statistic by normality, and the existence of the remainder term $R_n$ which obeys
$Pr\left\{|R_n|\geq \frac{\varepsilon}{n}\right\} = O(n^{-1})$, see (\ref{eqn:rn2}), and though nonzero is neglected in
the calculation. In Fig.~\ref{fig:rss_pdf}, we plot and compare the pdfs: (i) the overall shapes of the actual pdfs
(with the legend ``Simulation'') are quite similar to normality; (ii) the discrepancy in between reduces with $n$
increasing too. All those observations are consistent with and in turn demonstrate Theorem \ref{thm:rssconvergence}.

\begin{figure*}
\centering

\subfigure[Simulation Mean]{
\includegraphics[scale=0.8]{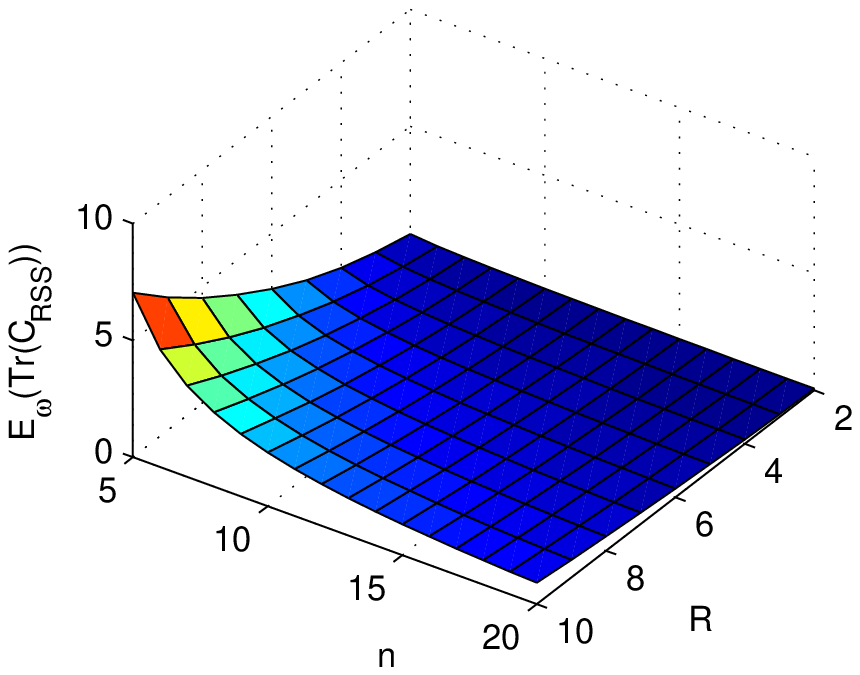}
\label{fig:rss_mean_s}}\subfigure[Simulation Standard Deviation]{
\includegraphics[scale=0.8]{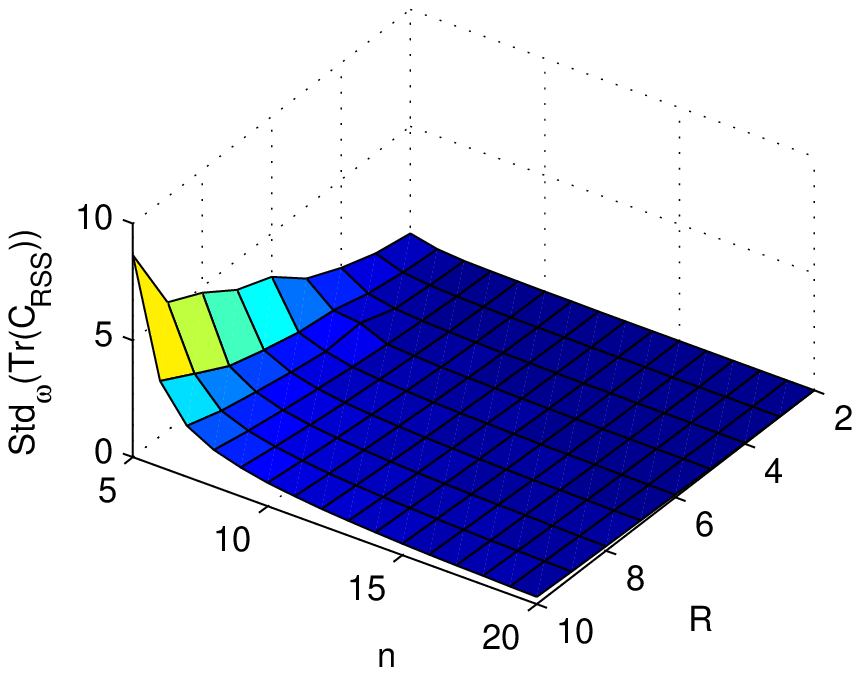}
\label{fig:rss_std_s}}

\subfigure[Analytical Mean]{
\includegraphics[scale=0.8]{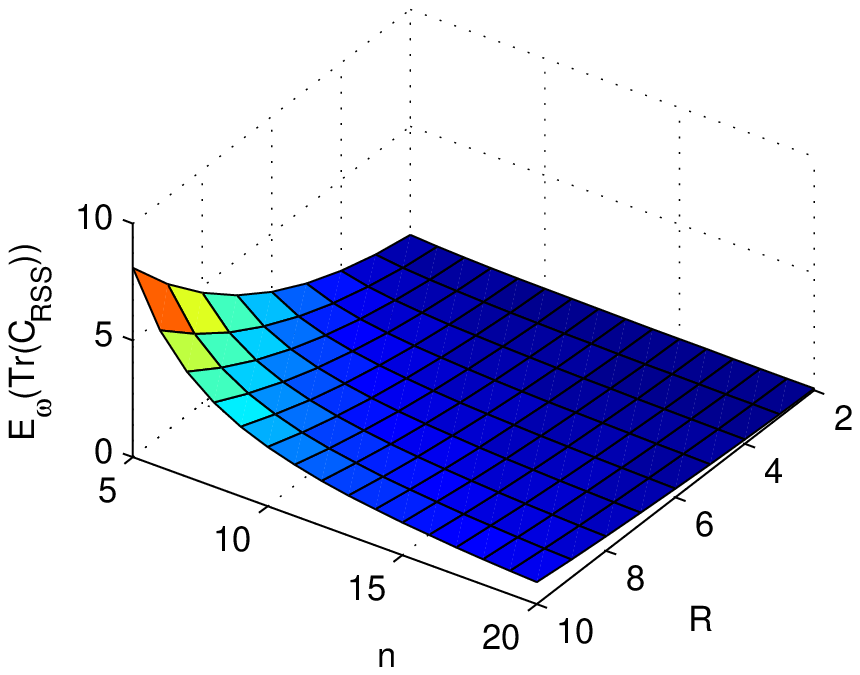}
\label{fig:rss_mean_a}}\subfigure[Analytical Standard Deviation]{
\includegraphics[scale=0.8]{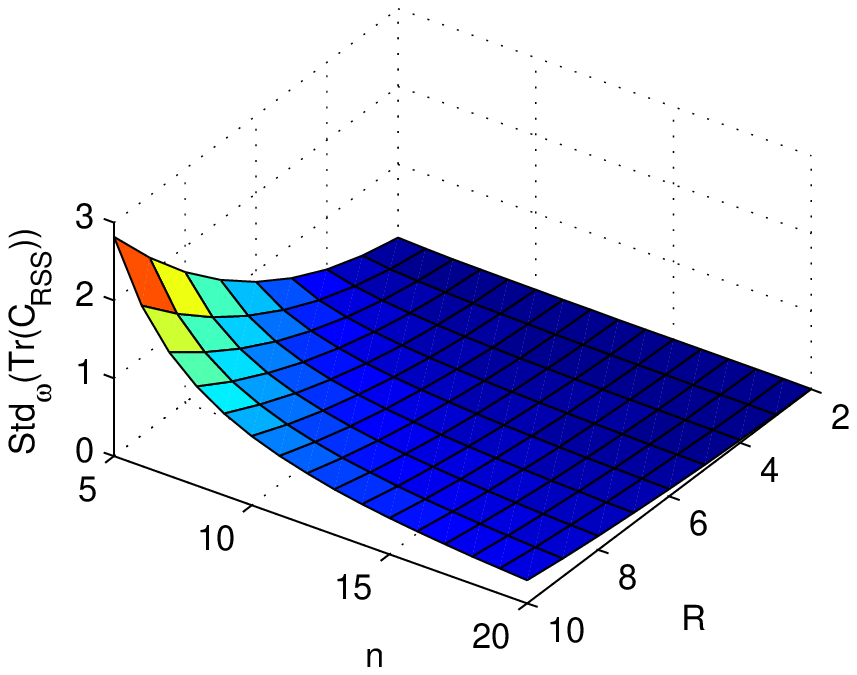}
\label{fig:rss_std_a}}

\subfigure[Mean]{
\includegraphics[scale=0.8]{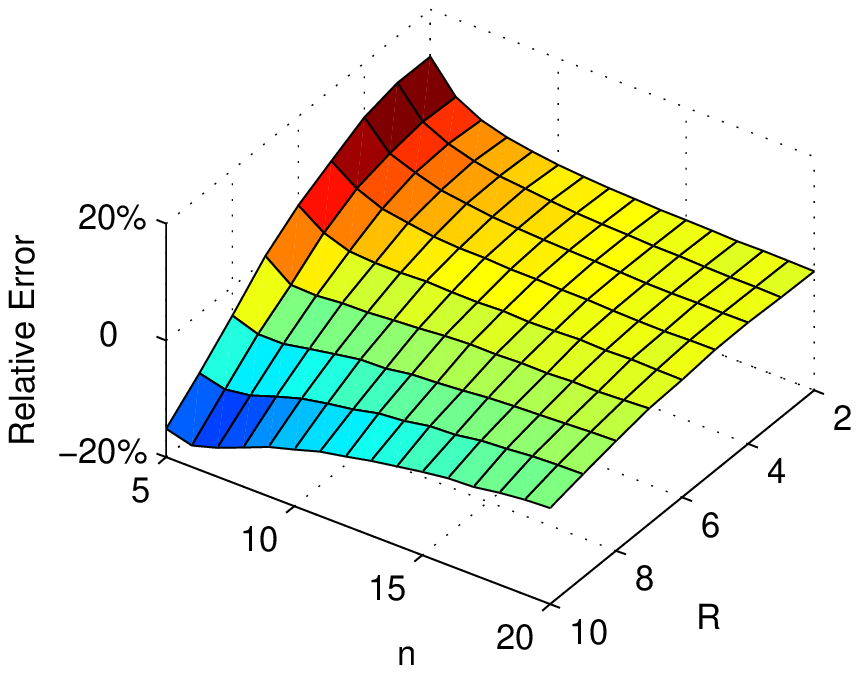}
\label{fig:rss_mean_re}}\subfigure[Standard deviation]{
\includegraphics[scale=0.8]{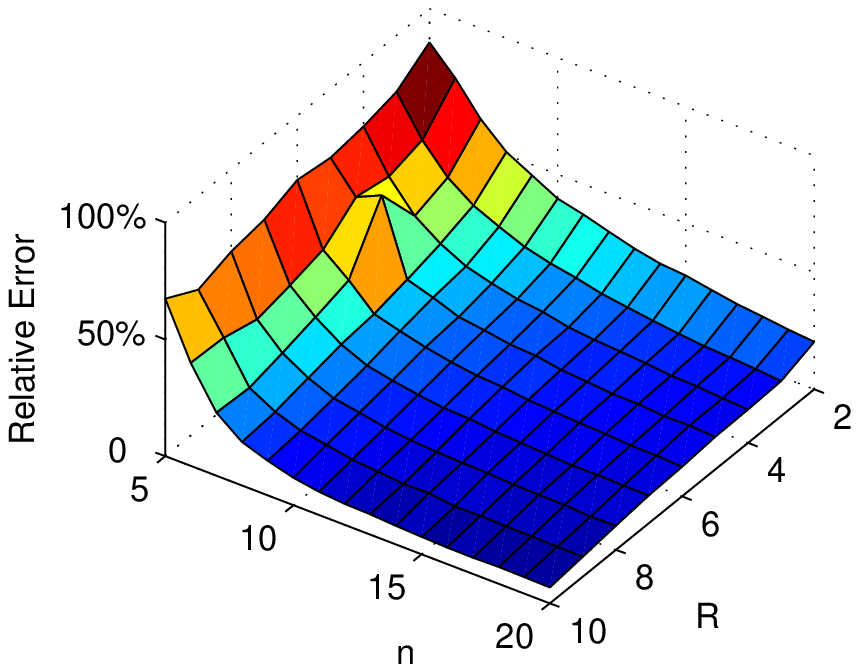}
\label{fig:rss_std_re}} \caption[]{The means and the standard deviations of $Tr(C_{RSS})$ from the simulations and the
formulas, and the corresponding relative errors, with $\sigma_{dB}=3.92, \alpha=2.3$ and $R_0=1$ m.} \label{fig:sim}
\end{figure*}

In the second place, we plot $E_{\mathbf\omega}(Tr(C_{RSS}))$ and $Std_{\mathbf\omega}(Tr(C_{RSS}))$ from both
simulations and the formulas (\ref{eqn:etrrss3}) and (\ref{eqn:stdtrrss2}) in Figs.~\ref{fig:rss_mean_s},
\ref{fig:rss_mean_a}, \ref{fig:rss_std_s} and \ref{fig:rss_std_a}, and evidently, the larger is $n$, the more accurate
are the formulas. Note that when $n=5$, the associated curve corresponding to the actual values of
$Std_{\mathbf\omega}(Tr(C_{RSS}))$ in Fig.~\ref{fig:rss_std_s} is non-smooth, and the most probable reason is that the
actual value of $Std_{\mathbf\omega}(Tr(C_{RSS}))$ is infinite when $n$ is as small as $5$. For better comparison, we
plot the \emph{relative errors}\footnote{Relative error is the ratio of the difference between the quantity from
simulations and that from a corresponding formula to the former one.} in Fig.~\ref{fig:rss_mean_re}
and~\ref{fig:rss_std_re}. It can be seen that: (i) $E_{\mathbf\omega}(Tr(C_{RSS}))$ is underestimated by
(\ref{eqn:etrrss3}) when $R$ is small, say $R=2$ m, but is overestimated by (\ref{eqn:etrrss3}) when $R$ is large, say
$R=10$ m, while the associated absolute value of the relative error decreases with $n$ increasing in most cases; (ii)
$Std_{\mathbf\omega}(Tr(C_{RSS}))$ is underestimated by (\ref{eqn:stdtrrss2}) and the associated absolute value of the
relative error decreases with increasing $n$ and $R$; (iii) suppose the absolute value of the relative error below
$10\%$ is acceptable: when $R=2$ m, (\ref{eqn:etrrss3}) is applicable if $n\geq6$, but (\ref{eqn:stdtrrss2}) is not
applicable even if $n=20$; when $R=10$ m, both (\ref{eqn:etrrss3}) and (\ref{eqn:stdtrrss2}) are applicable if
$n\geq11$.

In what follows, we present several useful remarks on the properties of sensor localization provided that
(\ref{eqn:etrrss3}) and (\ref{eqn:stdtrrss2}) are applicable. It is notable that in (\ref{eqn:etrrss3}) and
(\ref{eqn:stdtrrss2}), $E_{\mathbf\omega}(Tr(C_{RSS}))$ and $Std_{\mathbf\omega}(Tr(C_{RSS}))$ normalized by $R^2$ (or
$R_0^2$) are dependent upon the ratio $\frac{R}{R_0}$; hence, we simplify the discussion involving both $R_0$ and $R$
by letting $R_0=1$ m and only concentrating on $R$.

\begin{rem}
According to (\ref{eqn:etrrss3}),$E_{\mathbf\omega}(Tr(C_{RSS}))$ is in inverse proportion to $n$, and thus a
\emph{critical} value $n^*$ differing from the parameters $R_0, R, \sigma_{dB}$ and $\alpha$ can be determined, such
that having more anchors than $n^*$ contributes little to the quality of sensor localization.
\end{rem}

\begin{rem}
It can be easily deduced that both (\ref{eqn:etrrss3}) and (\ref{eqn:stdtrrss2}) monotonically decrease with $R$
decreasing, as illustrated in Fig.~\ref{fig:rss_mean_a} and~\ref{fig:rss_std_a}; the reason is that long distance
measurements from RSS suffer greater errors, and thus produce worse localization performance. Therefore, given a fixed
$n$, distance measurements from a sensor are better made at locations as close to the sensor as possible. Moreover, it
turns out that using more distance measurements spread over a wide range is not necessarily better than using fewer
distance measurements but spread in a narrow range in terms of $E_{\mathbf\omega}(Tr(C_{RSS}))$. For instance,
$E_{\mathbf\omega}(Tr(C_{RSS}))$ is approximately $0.52431$ m$^2$ given $n=15$ and $R=6$ m, but is around $0.43174$
m$^2$ given $n=10$ and $R=4$ m. Thus, tradeoff should be made between the number of the anchors (i.e. $n$) and their
spreading (i.e. $R_0$ and $R$) in sensor localization.
\end{rem}

\begin{rem}
Though we have discussed the positive and negative impacts of increasing $n$ and $R$ on localization performance
separately, the variables are correlated in some situations, and so the impacts are related. Normally, increasing all
the transmission powers in a wireless sensor network enlarges the communication coverage of every node, and both $n$
and $R$ tend to rise; finally, a positive impact is embodied, since $Tr(C_{RSS})$ and its mean will definitely decrease
according to \cite{PH03}.
\end{rem}

\begin{rem}
The dispersion of the distribution of $Tr(C_{RSS})$ reflects its sensitivity to sensor-anchor geometries. Specifically,
as illustrated in Fig.~\ref{fig:rss_pdf}, with a large dispersion, say $n=5$, the chance of having two different
sensor-anchor geometries to lead to a big difference in the resulting values of $Tr(C_{RSS})$ is large, implying a
large sensitivity, and we should be careful about sensor-anchor geometries; by contrast, with a small dispersion, say
$n\geq15$, the chance is certainly small, so is the sensitivity, and there is less reason to worry about sensor-anchor
geometries even if the anchors are randomly deployed. Given a random variable, the coefficient of variation, defined to
be the ratio of its standard deviation to its mean, is a normalized measure of dispersion of its distribution.
Accordingly, the coefficient associated with $Tr(C_{RSS})$ has the order of $O(n^{-\frac{1}{2}})$; the less is the
coefficient, the smaller is the sensitivity.

\end{rem}

\subsection{Expansion to Bearing-only Localization}\label{sec:bearing} A bearing is the angle between a north-south line
and a line connecting a sensor to an anchor, and is measured in a clockwise direction. In bearing-only localization,
bearing measurements associated with one sensor and at least two anchors (which are not collinear with the sensor) are
required to determine the sensor location. We still consider the sensor and $n$ ($n\geq2$) anchors, as shown in
Fig.~\ref{fig:network}. But in our study, the angles $\{\alpha_i, i=1,\cdots,n\}$ are assumed to be measured as
bearings, due to the fact that this set of measurements is equivalent to the set of real bearing measurements as far as
our study is concerned. Henceforth, we make the following assumption as is commonly used in the studies on bearing-only
localization \cite{BFADP10,Dogancay05}.
\begin{assumption}\label{asp:bearing1}
The bearing measurements $\{\hat{\alpha}_i, i=1,\cdots,n\}$ are statistically independent and Gaussian with means
$\{\alpha_i, i=1,\cdots,n\}$ and the same variance $\sigma_{\alpha}^2$.
\end{assumption}

In bearing-only localization, any anchor should not overlap with the sensor; otherwise, the corresponding bearing
measurement will be invalid. To model a random sensor-anchor geometry, we assume that the $n$ anchors are randomly and
uniformly distributed inside an annulus centered at the sensor and defined by radii $R_0$ and $R$ ($R_0<R$), as in
Assumption \ref{asp:rss1} in Section \ref{sec:formulation}.

Define $F_{B}$ to be the associated FIM in bearing-only localization, and according to \cite{BFADP10}, we have
\begin{equation}
F_{B}=\frac{1}{\sigma_{\alpha}^2}\left(
 \begin{array}{cc}
  \sum_{i=1}^n  \frac{\cos^2\alpha_i}{d_i^2} &  \sum_{i=1}^n  -\frac{\cos\alpha_i\sin\alpha_i}{d_i^2}\\
  \sum_{i=1}^n  -\frac{\cos\alpha_i\sin\alpha_i}{d_i^2}   &  \sum_{i=1}^n  \frac{\sin^2\alpha_i}{d_i^2}
 \end{array}
\right).
\end{equation}

\begin{rem}
Obviously, on replacing $\frac{1}{\sigma_{\alpha}^2}$ by $b$, $F_{B}$ will have the same form as the FIM in the RSS
case, i.e. $F_{RSS}$ in (\ref{eqn:frss}). Hence, all the conclusions about $Tr(C_{RSS})$ except those relevant to $b$
are still correct in bearing-only localization.
\end{rem}

\section{Results with TOA Measurements}\label{sec:toa}
Since multi-path fading due to reflection, diffraction and scattering of the radio signal causes variations in the
received power and severely reduces the accuracy of the RSS ranging approach, RSS-based localization normally achieves
low performance. By contrast, comparatively accurate distance estimates can be obtained using TOA measurements. In this
section, we extend our study in the TOA case.

\subsection{Problem Formulation in the TOA Case}
TOA refers to the travel time of a radio signal through a medium from a single transmitter to a remote single receiver.
By the relation between signal propagation speed in this medium, the time is a measure for the distance between
transmitter and receiver. Denote by $\{T_i, i=1,\cdots,n\}$ the measured TOA between the sensor and the $n$ anchors. As
in both theoretical studies \cite{CHC06,ZD10,XDD11} and experimental studies \cite{PH03} on TOA-based sensor
localization under line-of-sight conditions, we assume
\begin{assumption}\label{asp:toa2}
$\{T_i, i=1,\cdots,n\}$ are statistically independent and Gaussian with means $\{\frac{d_i}{c},i=1,\cdots,n\}$ ($c$ is
the speed of propagation) and same variance $\sigma_T^2$.
\end{assumption}

To model the random sensor-anchor geometry, we make the following assumption.
\begin{assumption}\label{asp:toa1}
The $n$ anchors are randomly and uniformly deployed within a circle of radius $R$ ($R>0$) centered at the sensor.
\end{assumption}

\begin{rem}
Because the TOA measurement model is still valid when the actual distance $d_i$ is $0$, we do not need to restrict
$d_i$ to be greater than $0$ as we do in the RSS case, and simply make Assumption \ref{asp:toa1}. It turns out that
$\{d_i, i=1,\cdots, n\}$ and $\{\alpha_i,i=1,\cdots,n\}$ are mutually independent; the only difference from the RSS
case is that the pdf of $d_i$ is $\frac{2d_i}{R^2}$.
\end{rem}

Similarly to the RSS case, we can obtain the FIM and the trace of the CRLB matrix in the TOA case as follows
\begin{eqnarray}
F_{TOA}  &=&  \frac{1}{\sigma_T^2c^2}\left(
 \begin{array}{cc}
  \sum_{i=1}^n  \cos^2\alpha_i &  \sum_{i=1}^n  \cos\alpha_i\sin\alpha_i\\
  \sum_{i=1}^n  \cos\alpha_i\sin\alpha_i   &  \sum_{i=1}^n  \sin^2\alpha_i \\
 \end{array}
\right), \\
Tr(C_{TOA})  &=&     \frac{ \sigma_T^2c^2 n }{ \sum_{1\leq i<j\leq n} \sin^2(\alpha_i-\alpha_j)  }.
\end{eqnarray}

Evidently, $Tr(C_{TOA})$ is a random variable and is independent of $\{d_i, i=1,\cdots,n\}$. Furthermore, the
sufficient conditions in Theorem \ref{thm:moments} also hold for the first and second moments of $Tr(C_{TOA})$. As
before, $Tr(C_{TOA})$ is used as the scalar metric for the performance limit in the TOA case.

\subsection{Theory}
By using Lemma \ref{lem:expansion}, we derive the following theorem.
\begin{thm}\label{thm:toaconvergence}
Let $\sigma_T$ and $Tr(C_{TOA})$ be the same variables as defined in the previous subsection and define a sequence of
random variables
\begin{equation}
V_n=\left(\frac{n(n-1)}{2\sigma_T^2c^2}\right)Tr(C_{TOA})-2n+2.
\end{equation}
Then, as $n\to\infty$, $V_n$ converges in distribution to a chi-square random variable of degree $2$.
\end{thm}
\begin{proof}
See Appendix \ref{sec:proofoftoaconvergence}
\end{proof}

\begin{rem}
According to Theorem \ref{thm:toaconvergence}, if $n$ is sufficiently large, the pdf of $Tr(C_{TOA})$ can be
approximated by
\begin{equation}\label{eqn:toa_approximate_distribution}
\frac{n(n-1)}{2\sigma_T^2c^2}f_{\chi}(\frac{n(n-1)}{2\sigma_T^2c^2}x-2n+2),
\end{equation}
where $f_{\chi}(\cdot)$ is the pdf of the chi-square random variable of degree $2$. Therefore, we can approximate the
moments of $Tr(C_{TOA})$ by the corresponding moments of the random variable defined by
(\ref{eqn:toa_approximate_distribution}), namely,
\begin{eqnarray}
E_{\mathbf\omega}(Tr(C_{TOA})) &\approx& \frac{4\sigma_T^2c^2}{n-1} \label{eqn:etrtoa2},\\
Std_{\mathbf\omega}(Tr(C_{TOA})) &\approx&
\frac{4\sigma_T^2c^2}{n(n-1)}\label{eqn:stdtrtoa2}.
\end{eqnarray}
\end{rem}

\subsection{Simulations}

In the simulations, we let $\sigma_Tc=1$ m and plot the distribution functions and pdfs of $Tr(C_{TOA})$ from both
simulations (with the legend ``Simulation'') and the random variable defined by
(\ref{eqn:toa_approximate_distribution}) (with the legend ``Formula'') for $n = 5, 10, 15, 20$ in
Fig.~\ref{fig:toadistribution}. It can be seen that the discrepancies between the pairs of distribution functions (or
pdfs) are not as obvious as in the RSS case, and also vanish with $n$ increasing, which is consistent with Theorem
\ref{thm:toaconvergence}. Note that every curve of the pdf from the simulations contains a sharp curvature, but the
corresponding curve from the given random variable is perfectly smooth; this does not invalidate Theorem
\ref{thm:toaconvergence}, due to the fact that as $n$ increases, the curve associated with $Tr(C_{TOA})$ tends to
become an impulse such that the minimum of $Tr(C_{TOA})$ will overlap with the value of $Tr(C_{TOA})$ attaining the
maximum of the pdf.

\begin{rem}
From Fig.~\ref{fig:toa_cdf}, we can obtain a clear understanding about $Tr(C_{TOA})$. When $n=5$, $Tr(C_{TOA})$ has a
lower limit around $0.8$ m$^2$ and can be far greater than $(\sigma_Tc)^2=1$ m$^2$ ($\sigma_Tc$ denotes the accuracy of
distance measurements from TOA); when $n\geq10$, the lower limit of $Tr(C_{TOA})$ reduces to be less than $0.4$ m$^2$
and $Tr(C_{TOA})$ could hardly take values above $1$ m$^2$.
\end{rem}

\begin{figure}
\centering \subfigure[]{
\includegraphics[scale=0.8]{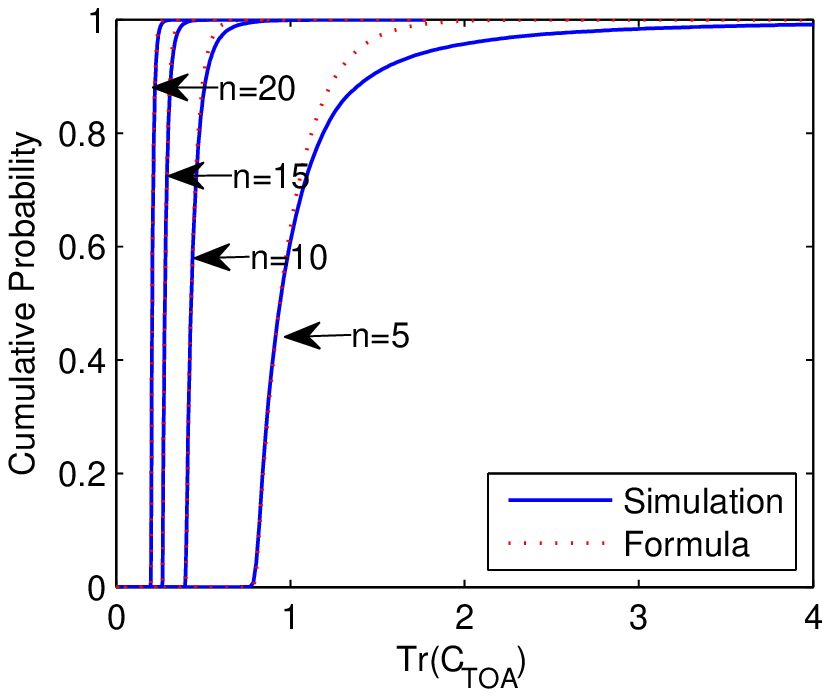}
\label{fig:toa_cdf}} \subfigure[]{
\includegraphics[scale=0.8]{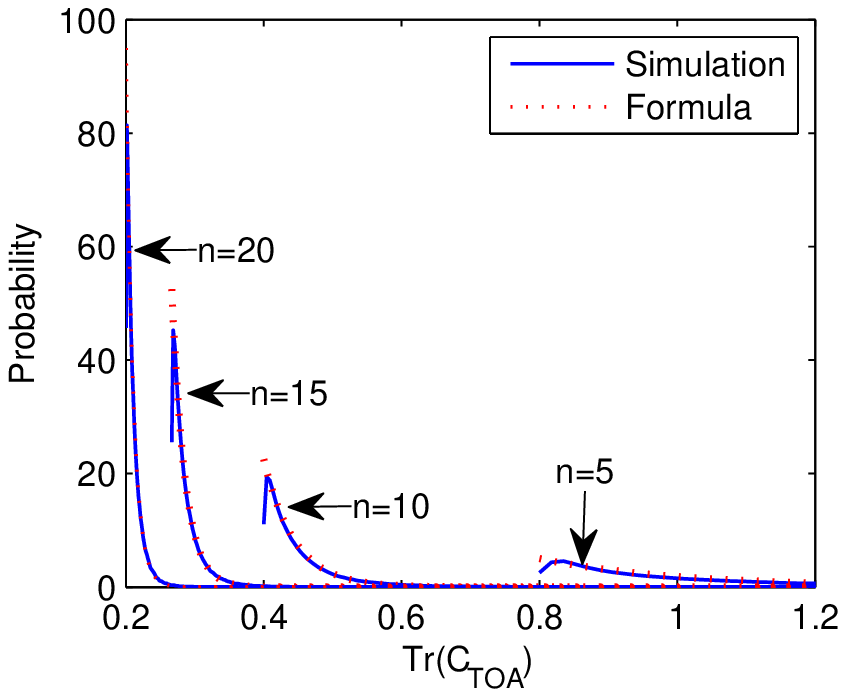}
\label{fig:toa_pdf}}  \caption[]{The distribution functions and pdfs of $Tr(C_{TOA})$ with $R=10$m and $\sigma_Tc=1$m.}
\label{fig:toadistribution}
\end{figure}

Next, we plot  $E_{\mathbf\omega}(Tr(C_{TOA}))$ and $Std_{\mathbf\omega}(Tr(C_{TOA}))$ from simulations and approximate
formulas (\ref{eqn:etrtoa2}) and (\ref{eqn:stdtrtoa2}) in Fig.~\ref{fig:toa_mean}, as well as the associated relative
errors in Fig.~\ref{fig:toa_re}. Although both formulas produce overestimates, the relative error of
(\ref{eqn:etrtoa2}) is much smaller than that of (\ref{eqn:stdtrtoa2}). Assuming that a relative error below $10\%$ is
acceptable, (\ref{eqn:etrtoa2}) is applicable if $n\geq6$, whereas (\ref{eqn:stdtrtoa2}) is not applicable even if
$n=20$.

\begin{figure}
\centering \subfigure[]{
\includegraphics[scale=0.8]{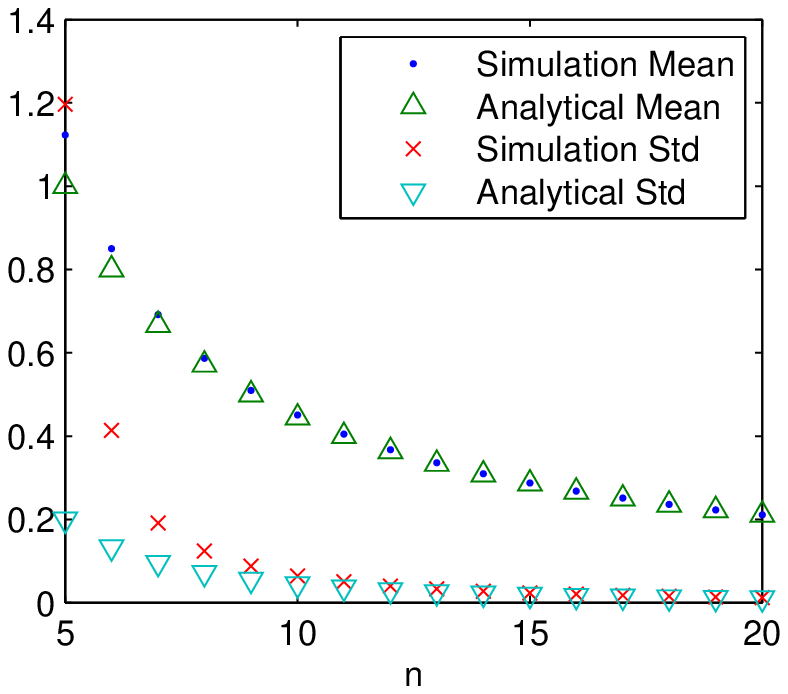}
\label{fig:toa_mean}}  \subfigure[]{
\includegraphics[scale=0.8]{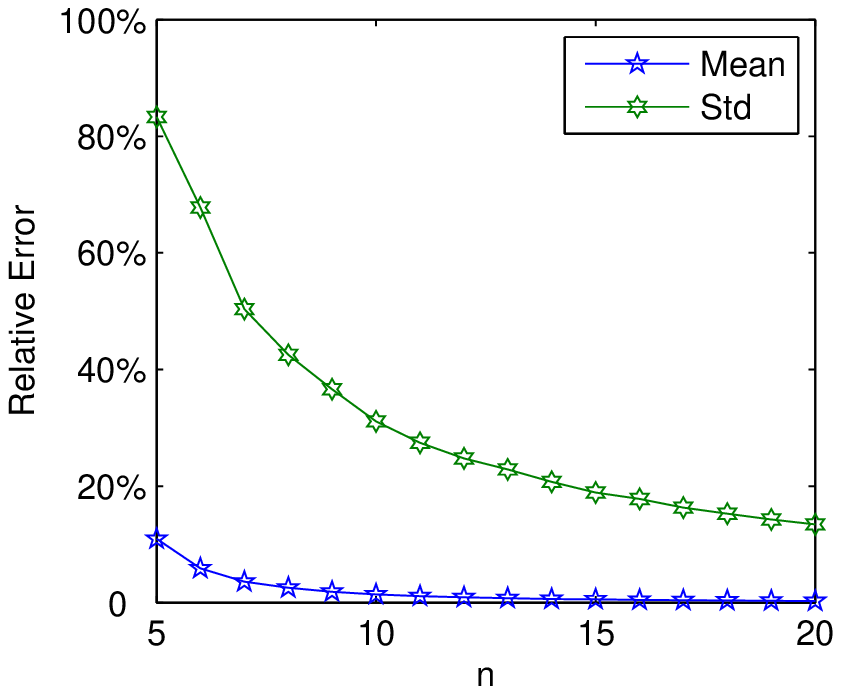}
\label{fig:toa_re}}\caption[]{The means and  the standard deviations of $Tr(C_{TOA})$ from the simulations and
formulas, and the corresponding relative errors, with $R=10$m and $\sigma_Tc=1$m.} \label{fig:toa}
\end{figure}

\begin{rem}
Since (\ref{eqn:etrtoa2}) is extraordinarily accurate and shows that $E_{\mathbf\omega}(Tr(C_{TOA}))$ is inversely
proportional to $n$, we can determine a critical value $n^*$ differing from $\sigma_Tc$, such that having more anchors
than $n^*$ contributes little to the quality of sensor localization, as in the RSS case. Further, if $n$ is
sufficiently large such that (\ref{eqn:stdtrtoa2}) is applicable, the coefficient of variation for $Tr(C_{TOA})$
approximately equals $\frac{1}{n}$, implying that $Tr(C_{TOA})$ reduces to $0$ with a faster rate than $Tr(C_{RSS})$ as
$n$ increases. 
\end{rem}

\section{Conclusion and Future Work}\label{sec:conclusion}
In this paper, we investigated the performance limit of single-hop sensor localization using RSS, TOA or bearing
measurements by statistical sensor-anchor geometry modeling. That is, the anchor locations are assumed to be random,
and the scalar metric for the performance limit of sensor localization, i.e. the trace of the associated CRLB matrix,
consequently becomes random. We came up with formulas expressing the asymptotic behavior of the scalar metric in terms
of distribution, mean and standard deviation. Specifically, as the number of the anchors goes to infinity, the scalar
metric in the RSS/bearing case is asymptotically normal and its rate of convergence to normality was also derived; in
the TOA case, the scalar metric converges to a random variable which is an affine transformation of a chi-square random
variable of degree $2$; we presented approximate formulas for means and standard deviations of the scalar metric in
both the RSS/bearing case and the TOA case. Although these formulas are asymptotic in the number of the anchors, they
turn out to be remarkably accurate in predicting the performance limit of sensor localization in many cases, even if
the number of the anchors is fairly small. In addition, we analyzed the formulas to demonstrate some general properties
of sensor localization and carried out extensive simulations to verify the conclusions.

Considering the similarities between the models for bearing measurements and angle of arrival (AOA) measurements, we
can easily expand the conclusions in the RSS/bearing case to AOA-based localization. Furthermore, distance measurements
in range-only localization are often modeled to be mutually independent and Gaussian \cite{BFADP10}, which is the same
as occurs with TOA measurements, and thus, it is straightforward to expand the conclusions in the TOA case to
range-only localization. In future work, we may expand our study into three-dimensional space and multi-hop sensor
localization.


%

\appendices

\section{Derivation of $Tr(C_{RSS})$}\label{sec:trace}
According to the formulation of $F_{RSS}$, we have
\begin{equation}
C_{RSS}=\frac{1}{b}\left(
 \begin{array}{cc}
  \sum_{i=1}^n  \frac{\cos^2\alpha_i}{d_i^2} &  \sum_{i=1}^n  \frac{\cos\alpha_i\sin\alpha_i}{d_i^2}\\
  \sum_{i=1}^n  \frac{\cos\alpha_i\sin\alpha_i}{d_i^2}   &  \sum_{i=1}^n  \frac{\sin^2\alpha_i}{d_i^2} \\
 \end{array}
\right)^{-1}.
\end{equation}

Then,
\begin{eqnarray}
 Tr(C_{RSS})  &=&\frac{\frac{1}{b}\sum_{i=1}^n  \frac{1}{d_i^2} }{\left(\sum_{i=1}^n  \frac{\cos^2\alpha_i}{d_i^2} \right) \left(\sum_{i=1}^n  \frac{\sin^2\alpha_i}{d_i^2}\right) - \left(\sum_{i=1}^n \frac{\cos\alpha_i\sin\alpha_i}{d_i^2}\right)^2 }\nonumber\\
 &=&\frac{\frac{1}{b}\sum_{i=1}^n  \frac{1}{d_i^2} }{\sum_{1\leq i,j\leq n} \frac{\cos\alpha_i\sin\alpha_j\sin(\alpha_i-\alpha_j)}{d_i^2d_j^2}  }\nonumber\\
 &=&\frac{\frac{1}{b}\sum_{i=1}^n  \frac{1}{d_i^2} }{\sum_{1\leq i<j\leq n} \frac{(\cos\alpha_i\sin\alpha_j\sin(\alpha_i-\alpha_j)+\cos\alpha_j\sin\alpha_i\sin(\alpha_j-\alpha_i))}{d_i^2d_j^2} }\nonumber\\
 &=&\frac{\frac{1}{b}\sum_{i=1}^n  \frac{1}{d_i^2} }{\sum_{1\leq i<j\leq n} \frac{\sin^2(\alpha_i-\alpha_j)}{d_i^2d_j^2}  }.\nonumber
\end{eqnarray}

\section{Proof of Theorem \ref{thm:moments}}\label{sec:moments}
\begin{proof}
Let the denominator of $Tr(C_{RSS})$ be
\begin{equation}
Y_n= \sum_{1\leq i<j\leq n} \left(\frac{\sin^2(\alpha_i-\alpha_j)}{d_i^2d_j^2}\right).
\end{equation}
Since the numerator of $Tr(C_{RSS})$ is obviously bounded, once the first and second moments of $Y_n^{-1}$ are finite,
the corresponding moments of $Tr(C_{RSS})$ are finite too. Supposing the pdf of a random variable $X$, denoted $f(x)$,
is continuous on $[0,+\infty)$, if $f(0)=0$ and $f'(0)<\infty$, then $E(X^{-1})<\infty$ from \cite{PC85}. We let
$f_n(x)$ be the pdf of $Y_n$. To prove $E(Y_n^{-1})<\infty$, we first show $f_n(0)=0$ given $n\geq4$ as follows.

\begin{eqnarray*}
&&  f_n(0)= \lim_{\varepsilon\to0}\frac{Pr\{Y_n\leq\varepsilon\}-Pr\{Y_n\leq0\}}{\varepsilon} \\
&=& \lim_{\varepsilon\to0} \frac{1}{\varepsilon}{\int\cdots\int}_{\frac{2}{n(n-1)}\sum_{1\leq i<j\leq n}\left(\frac{\sin^2(\alpha_i-\alpha_j)}{d_i^2d_j^2}\right)\leq\varepsilon} f_{\alpha}(\alpha_1)\cdots f_{\alpha}(\alpha_n) f_d(d_1) \cdots f_d(d_n) d \alpha_1\cdots d{\alpha_n} d d_1 \cdots d d_n \\
&\leq& \frac{1}{(2\pi )^n}\lim_{\varepsilon\to0} \frac{1 }{\varepsilon}\int\cdots\int_{\frac{2}{n(n-1)}\sum_{1\leq
i<j\leq 4}\left(\frac{\sin^2(\alpha_i-\alpha_j)}{d_i^2d_j^2}\right)\leq\varepsilon} d \alpha_1\cdots
d{\alpha_n}\\
&\leq& \frac{1}{(2\pi )^4}\lim_{\varepsilon\to0} \frac{1}{\varepsilon}\int\cdots\int_{\sum_{1\leq i<j\leq
4}\left(\sin^2(\alpha_i-\alpha_j)\right)\leq\frac{n(n-1)R^4\varepsilon}{2}} d \alpha_1 d\alpha_2d\alpha_3d\alpha_4 \\
&\leq& \frac{1}{(2\pi )^4}\lim_{\varepsilon\to0}
\frac{1}{\varepsilon}\int_{-\pi<\alpha_1\leq\pi}\int_{|\alpha_2-\alpha_1+k\pi|\leq\sqrt{\frac{n(n-1)R^4\varepsilon}{2}}}\int_{|\alpha_3-\alpha_2+k\pi|\leq\sqrt{\frac{n(n-1)R^4\varepsilon}{2}}}
\int_{|\alpha_4-\alpha_3+k\pi|\leq\sqrt{\frac{n(n-1)R^4\varepsilon}{2}}}  d \alpha_4\cr
&& d\alpha_3d\alpha_2d\alpha_1, \  k=-2,-1,0,1,2 \\
&=& \lim_{\varepsilon\to0} O(\varepsilon^{\frac{1}{2}\times3-1})=0.
\end{eqnarray*}

Next, we need to prove $f_n'(0)<\infty$, namely
\begin{equation}
  f_n'(0) 
  =\lim_{\xi\to0}\lim_{\varepsilon\to0}\frac{1}{\varepsilon\xi}\int\cdots\int_{\xi<Y_n\leq\xi+\varepsilon}f_{\alpha}(\alpha_1)\cdots f_{\alpha}(\alpha_n) f_d(d_1) \cdots f_d(d_n)d\alpha_1\cdots d\alpha_n dd_1
  \cdots dd_n<\infty.\label{eqn:accumulative}
\end{equation}

Because $f_n(\xi)$ is finite for any $\xi$, the limit calculus results in
\begin{equation}
  f_n'(0)=\lim_{\xi\to0,\varepsilon\to0}\frac{1}{\varepsilon\xi}\int\cdots\int_{\xi<Y_n\leq\xi+\varepsilon}f_{\alpha}(\alpha_1)\cdots f_{\alpha}(\alpha_n) f_d(d_1) \cdots f_d(d_n)d\alpha_1\cdots d\alpha_n dd_1
  \cdots dd_n.
\end{equation}
On replacing $\xi$ and $\varepsilon$ by $r\cos\theta$ and $r\sin\theta$, one has
\begin{eqnarray*}
  &&f_n'(0)\\
  &=&\lim_{r\to0}\frac{1}{r^2 \cos\theta\sin\theta}\int\cdots\int_{r\cos\theta<Y_n\leq
r(\cos\theta+\sin\theta)}f_{\alpha}(\alpha_1)\cdots
f_{\alpha}(\alpha_n) f_d(d_1) \cdots f_d(d_n)d\alpha_1\cdots d\alpha_n dd_1 \cdots dd_n\\
&\leq& \frac{1}{(2\pi)^n}\lim_{r\to0}\frac{1}{r^2 \cos\theta\sin\theta}\int\cdots\int_{Y_n\leq
r(\cos\theta+\sin\theta)}dd_1\cdots dd_n.
\end{eqnarray*}
The term on the right hand side (RHS) of the above inequality can be treated similarly to proving $f_n(0)=0$, with the
result that for $n\geq5$,
\begin{equation}
\lim_{r\to0}\frac{1}{\varepsilon\xi}\int\cdots\int_{r\cos\theta<Y_n\leq r(\cos\theta+\sin\theta)}d\alpha_1\cdots
d\alpha_n<\infty
\end{equation}
and thus $E(Y_n^{-1})<\infty$. It is straightforward to extend the result about the first moment to the second moment
by following a similar line of argument, and then the theorem is proved.
\end{proof}

\section{Proof of Lemma \ref{lem:expansion}}\label{sec:proofofexpansion}

Prior to the proof, we review some background about $U$-statistics. Firstly, given a $U$-statistic $U_n$ with kernel
$h(x_1,\cdots,x_r)$, it can be rewritten using Hoeffding's (or $H$-) Decomposition as follows (see \cite{Lee90})
\begin{equation}
   U_{n}=\sum_{j=1}^r\frac{j!(n-j)!}{n!} \sum_{1\leq i_1<\cdots<i_j\leq n}\lambda_j(X_{i_1},\cdots,X_{i_j})
\end{equation}
where
\begin{eqnarray*}
   \lambda_1(x_1)&=&E(h(x_1,X_2,\cdots,X_r)),\\
   \lambda_2(x_1,x_2)&=&E(h(x_1,x_2,X_3,\cdots,X_r))-\lambda_1(x_1)-\lambda_1(x_2),\\
   &\cdots&\\
   \lambda_r(x_1\cdots,x_r)&=&h(x_1,\cdots,x_r)-\sum_{j=1}^{r-1}\sum_{1\leq i_1<\cdots<i_j\leq
   r}\lambda_j(x_{i_1},\cdots,x_{i_j}).
\end{eqnarray*}

We define
\begin{equation}
 \Lambda_k = \sum_{1\leq i_1<\cdots<i_k\leq n}\lambda_k(X_{i_1},\cdots,X_{i_k})
\end{equation}
and have
\begin{equation}\label{eqn:transform1}
 \sum_{1\leq i_1<\cdots<i_r\leq n}h(X_{i_1},\cdots,X_{i_r}) = \sum_{k=1}^r\left( \frac{(n-k)!}{(n-r)!(r-k)!}
   \Lambda_k\right).
\end{equation}

Then, the following lemma concerning $\Lambda_k$ can be derived.
\begin{lem}\label{lem:l1}
(Lemma $A1$ in \cite{Maesono98}) For any given $q\geq 2$, if $E|h(X_1,\cdots,X_r)|^q<\infty$, then there exists a
positive constant $a$, which may depend on $h(X_1,\cdots,X_r)$ and the distribution of $X_1,\cdots,X_r$ but is
independent of $n$, such that
\begin{equation}
  E|\Lambda_k|^q\leq an^{\frac{qk}{2}}.
\end{equation}
\end{lem}

In addition, we can obtain the critical lemma as follows.
\begin{lem}\label{lem:l2}
Let $r$ be a fixed positive integer. For any given $q\geq2$, if $E|h(X_1,\cdots,X_r)|^q<\infty$, then there exists a
constant $a$ independent of $n$, such that
\begin{equation}
E\left\{ \left|n^{-r}\sum_{1\leq i_1<\cdots <i_r\leq n}h(X_{i_1},\cdots,X_{i_r})\right|^q\right\} \leq a
n^{-\frac{q}{2}}.
\end{equation}

\end{lem}

\begin{proof}
From (\ref{eqn:transform1}) and $c_r$-inequality, we get
\begin{equation}
E \left(\left|n^{-r }\sum_{1\leq i_1<\cdots <i_r\leq
   n}h(X_{i_1},\cdots,X_{i_r})\right|^q\right)
   \leq  r^{q-1}\sum_{k=1}^r E \left(\left|n^{-r }\left( \frac{(n-k)!}{(n-r)!(r-k)!}
   \Lambda_k\right)\right|^q\right).
\end{equation}
Noting that $n^{-r}\left( \frac{(n-k)!}{(n-r)!(r-k)!}\right) =O(n^{-k})$, by Lemma \ref{lem:l1}, we know that there
exist constants $a_1,\cdots,a_r$, which are all independent of $n$, such that
\begin{equation}
\sum_{k=1}^r E\left( \left|n^{-r}\left( \frac{(n-k)!}{(n-r)!(r-k)!}
   \Lambda_k\right)  \right|^q\right)
   \leq \sum_{k=1}^r(a_k n^{-kq+\frac{kq}{2}} ).
\end{equation}
Thus, by letting $a=r\max\{a_1,\cdots,a_r\}$, we have
\begin{equation}
  E \left( \left|n^{-r}\sum_{1\leq i_1<\cdots <i_r\leq
   n}h(X_{i_1},\cdots,X_{i_r})\right|^q\right) \leq an^{-\frac{q}{2}}.
\end{equation}
\end{proof}


The proof below for Lemma \ref{lem:expansion} is in line with that of Theorem $1$ in \cite{Maesono98}, but Lemma
\ref{lem:expansion} enhances the statement about the remainder term $R_n$ by 
(\ref{eqn:rn2}) and (\ref{eqn:rn3}) which are key to prove the subsequent theorems.
\begin{proof}
For ease of presentation, we say a random variable satisfies the Condition $\mathcal{A}$ if and only if it satisfies
the same two conditions (i.e. (\ref{eqn:rn2}) and (\ref{eqn:rn3})) as $R_n$ does in Lemma \ref{lem:expansion}. The
skeleton of the proof is: \emph{Step 1}, applying Taylor expansions on $S_n^{-1}$; {\em Step 2\&3}, identifying the
resulting terms in the expansion of $S_n^{-1}$ which satisfy the Condition $\mathcal{A}$; {\em Step 4}, rewriting
$S_n^{-1}$ by combining the terms satisfying the Condition $\mathcal{A}$ into one term; {\em Step 5}, multiplying
$S_n^{-1}$ and $T_n$, identifying the resulting terms which satisfy the Condition $\mathcal{A}$ and combining them into
the remainder term $R_n$.

\subsubsection*{Step 1}Let $\varepsilon$ and $q$ be any given positive real numbers and $q>2$. The $H$-decomposition of $S_n$ is
\begin{equation}\label{eqn:sn}
   S_n = m_s+\frac{2}{n}\sum_{i=1}^n\zeta_1(X_i)+\frac{2}{n(n-1)}\sum_{1\leq i<j\leq n}\zeta_2(X_i,X_j)
\end{equation}
where $m_s = m_1^2m_2$ and
\begin{eqnarray}
 \zeta_1(X_i) &=& (X^{(1)}_i-m_1)m_1m_2, \\
 \zeta_2(X_i,X_j) &=&
 X^{(1)}_iX^{(1)}_j\sin^2(X^{(2)}_i-X^{(1)}_j)-(X^{(1)}_i+X^{(1)}_j)m_1m_2+m_1^2m_2.
\end{eqnarray}
Applying Taylor expansions on $S_n^{-1}$ around $m_s$, we obtain
\begin{equation}\label{eqn:sninverse}
   S_n^{-1}=m_s^{-1}-m_s^{-2}(S_n-m_s)+m_s^{-3}(S_n-m_s)^2-(m_s+\vartheta_n)^{-4}(S_n-m_s)^3
\end{equation}
where $0\leq|\vartheta_n|\leq|S_n-m_s|$. We shall identify the terms on the RHS of (\ref{eqn:sninverse}) satisfying the
Condition $\mathcal{A}$.

\subsubsection*{Step 2} Firstly, consider $|m_s+\vartheta_n|^{-4}|S_n-m_s|^3$. It follows from (\ref{eqn:sn}) and
$c_r$-inequality that for any $p\geq1/3$,
\begin{equation}\label{eqn:snms3}
   \left|S_n-m_s\right|^{3p}\leq
   2^{3p-1} \left|\frac{2}{n}\sum_{i=1}^n\zeta_1(X_i)\right|^{3p}+2^{3p-1} \left|\frac{2}{n(n-1)}\sum_{1\leq i<j\leq
   n}\zeta_2(X_i,X_j)\right|^{3p}.
\end{equation}
For the first term on the RHS of (\ref{eqn:snms3}), from Lemma \ref{lem:l2} and Markov's inequality, we have
\begin{eqnarray}
Pr\left\{n\left|\frac{2}{n}\sum_{i=1}^n\zeta_1(X_i)\right|^{3}\geq\varepsilon\right\}  &=&  O(n^{-\frac{q}{2}}), \\
 Pr\left\{n(\ln n)\left|\frac{2}{n}\sum_{i=1}^n\zeta_1(X_i)\right|^{3}\geq\varepsilon\right\}  &=&  o(1). \label{eqn:lneq1}
\end{eqnarray}

Hence, the first term on the RHS of (\ref{eqn:snms3}) satisfies the Condition $\mathcal{A}$. Similarly, we can prove
that the second term satisfies the Condition $\mathcal{A}$ too. Furthermore, from (\ref{eqn:lneq1}), the second term on
the RHS of (\ref{eqn:snms3}) satisfying the Condition $\mathcal{A}$, and
\begin{eqnarray}
 &&Pr\left\{n(\ln n) |S_n-m_s|^3\geq \varepsilon \right\} \nonumber\\
 &\leq& Pr\left\{2\left(n(\ln n)\left|\frac{2}{n}\sum_{i=1}^n\zeta_1(X_i)\right|^3+n(\ln n)\left|\frac{2}{n(n-1)}\sum_{1\leq i<j\leq n}\zeta_2(X_i,X_j)\right|^3\right)\geq\varepsilon\right\} \nonumber\\
 &\leq& Pr\left\{2n(\ln n)\left|\frac{2}{n}\sum_{i=1}^n\zeta_1(X_i)\right|^3\geq\frac{\varepsilon}{2}\right\}+Pr\left\{2n(\ln n)\left|\frac{2}{n(n-1)}\sum_{1\leq i<j\leq n}\zeta_2(X_i,X_j)\right|^3\geq\frac{\varepsilon}{2} \right\}, \nonumber
\end{eqnarray}
we have
\begin{equation}\label{eqn:lneq3}
   Pr\left\{n(\ln n) |S_n-m_s|^3\geq \varepsilon \right\}=o(1).
\end{equation}
Similarly, we have
\begin{eqnarray}
Pr\left\{n |S_n-m_s|^3\geq \varepsilon \right\}=   O(n^{-\frac{q}{2}}),
\end{eqnarray}
and thus, $|S_n-m_s|^3$ satisfies the Condition $\mathcal{A}$.

Based on Lemma \ref{lem:l2} and Markov's inequality, we obtain
\begin{eqnarray}
Pr\left\{\left|\frac{2}{n}\sum_{i=1}^n\zeta_1(X_i)\right|\geq\varepsilon\right\}  &= & O(n^{-\frac{q}{2}}),\\
Pr\left\{\left|\frac{2}{n(n-1)}\sum_{1\leq i<j\leq n}\zeta_2(X_i,X_j)\right|\geq\varepsilon\right\}  &= & O(n^{-\frac{
q}{2}}).
\end{eqnarray}
This gives
\begin{equation}\label{eqn:psnms}
Pr\{|S_n-m_s|\geq \varepsilon\} = O(n^{-\frac{q}{2}}),
\end{equation}
and noting that
\begin{eqnarray}
 Pr\left\{\left|m_s+\vartheta_n\right|<\frac{m_s}{2}\right\} 
  &=& Pr\left\{-\frac{3m_s}{2}< \vartheta_n <-\frac{m_s}{2}\right\}  \nonumber\\
   &\leq& Pr\left\{|\vartheta_n| \geq \frac{m_s}{2}\right\}\nonumber\\
   &\leq& Pr\left\{|S_n-m_s| \geq \frac{m_s}{2}\right\},
\end{eqnarray}
we have
\begin{equation}\label{eqn:eq1}
   Pr\left\{\left|m_s+\vartheta_n\right|<\frac{m_s}{2}\right\} = O(n^{-\frac{q}{2}}).
\end{equation}
Furthermore, by (\ref{eqn:lneq3}), (\ref{eqn:eq1}) and
\begin{eqnarray}
&&Pr\left\{n(\ln n)|S_n-m_s|^3|m_s+\vartheta_n|^{-4}\geq \varepsilon \right\}\\
&=&Pr\left\{\{n(\ln n)|S_n-m_s|^3|m_s+\vartheta_n|^{-4}\geq\varepsilon\}\bigcap
\{|m_s+\vartheta_n|\geq\frac{m_s}{2}\}\right\}\nonumber\\
&&+Pr\left\{\{n(\ln n)|S_n-m_s|^3|m_s+\vartheta_n|^{-4}\geq\varepsilon\}\bigcap\{|m_s+\vartheta_n|<\frac{m_s}{2}\}\right\}\nonumber\\
&\leq& Pr\left\{n(\ln n)|S_n-m_s|^3\left(\frac{m_s}{2}\right)^{-4}\geq\varepsilon\right\}
+Pr\left\{|m_s+\vartheta_n|<\frac{m_s}{2}\right\},
\end{eqnarray}
we have
\begin{equation}
Pr\{n(\ln n)|S_n-m_s|^3|m_s+\vartheta_n|^{-4}\}  = o(1).
\end{equation}
Similarly, we have
\begin{equation}
Pr\left\{n|S_n-m_s|^3|m_s+\vartheta_n|^{-4}\geq\varepsilon\right\}=
O(n^{-\frac{q}{2}}).
\end{equation}
Hence, we conclude that $|S_n-m_s|^3|m_s+\vartheta_n|^{-4}$ {\em satisfies} the Condition $\mathcal{A}$.

\subsubsection*{Step 3}Secondly, we deal with the term $(S_n-m_s)^2$:
\begin{eqnarray}
(S_n-m_s)^2&=&\frac{4}{n}E|\zeta_1(X_i)|^2+\frac{4}{n^2}\sum_{i=1}^n\left(\zeta_1(X_i)^2-E|\zeta_1(X_i)|^2\right)+\frac{8}{n^2}\sum_{1\leq
i<j\leq n}\left(\zeta_1(X_i)\zeta_1(X_j)\right)\cr &&+\frac{8}{n^2(n-1)}\sum_{i=1}^n\zeta_1(X_i)\sum_{1\leq i<j\leq
n}\zeta_2(X_i,X_j). \label{eqn:sn2}
\end{eqnarray}
In the second term on the RHS of (\ref{eqn:sn2}), $\zeta_1(X_i)^2-E|\zeta_1(X_i)|^2$ can be regarded as the kernel of a
$U$-statistic; then according to Lemma \ref{lem:l2} and Markov's inequality, we have
\begin{eqnarray}
Pr\left\{n\left|\frac{4}{n^2}\sum_{i=1}^n\left(\zeta_1(X_i)^2-E|\zeta_1(X_i)|^2\right)\right|\geq \varepsilon \right\}   &=&  O(n^{-\frac{q}{2}}), \\
Pr\left\{n(\ln n)\left|\frac{4}{n^2}\sum_{i=1}^n\left(\zeta_1(X_i)^2-E|\zeta_1(X_i)|^2\right)\right|\geq \varepsilon \right\}   &=&
o(1),
\end{eqnarray}
and thus, the second term satisfies the Condition $\mathcal{A}$. For the third term on the RHS of (\ref{eqn:sn2}), we
have
\begin{equation}
 \frac{8 m_s^{-3}}{n^2}\sum_{1\leq i<j\leq n} \zeta_1(X_i)\zeta_1(X_j)
 =  \frac{8 m_s^{-3}}{n(n-1)}\sum_{1\leq i<j\leq n} \zeta_1(X_i)\zeta_1(X_j)
 -\frac{8 m_s^{-3}}{n^2(n-1)}\sum_{1\leq i<j\leq n}
 \zeta_1(X_i)\zeta_1(X_j).
\end{equation}
Based on Lemma \ref{lem:l2}, we can instantly obtain that the second term on the RHS of the above equation satisfies
the Condition $\mathcal{A}$. Define
\begin{eqnarray}
 \Lambda'_1 &=& \sum_{i=1}^n\zeta_1(X_i), \\
 \Lambda'_2 &=& \sum_{1\leq i<j\leq n}\zeta_2(X_i,X_j).
\end{eqnarray}
From Lemma \ref{lem:l1}, we can derive
\begin{eqnarray}
E\left|\frac{8}{n^2(n-1)}\sum_{i=1}^n\zeta_1(X_i)\sum_{1\leq
<i<j\leq n}\zeta_2(X_i,X_j)\right|^q
& =&O\left(n^{-\frac{3}{2}q}\right), \label{eqn:category2_2}
\end{eqnarray}
and thus, the last term on the RHS of (\ref{eqn:sn2})  satisfies the Condition $\mathcal{A}$.

\subsubsection*{Step 4}
We rewrite (\ref{eqn:sninverse}) by combining all the terms which satisfy the Condition $\mathcal{A}$ into a new term
$R_n'$ (which, as a consequence, satisfies the Condition $\mathcal{A}$ too), and obtain
\begin{eqnarray}
   S_n^{-1}&=&m_s^{-1}+ \frac{4m_s^{-3}}{n}E|\zeta_1(X_i)|^2-\frac{2m_s^{-2}}{n}
   \sum_{i=1}^n\zeta_1(X_i)-\frac{2m_s^{-2}}{n(n-1)}\sum_{1\leq
i<j\leq n}\zeta_2(X_i,X_j)\cr
   &&+\frac{8m_s^{-3}}{n(n-1)}\sum_{1\leq
i<j\leq n}\left(\zeta_1(X_i)\zeta_1(X_j)\right)+R'_n. \label{eqn:sninverse2}
\end{eqnarray}

\subsubsection*{Step 5}
Similarly to $S_n$, $T_n$ can be formulated using $H$-Decomposition as follows
\begin{equation}\label{eqn:tn}
   T_n=m_t+\frac{1}{n}\sum_{i=1}^n \tau_1(X_i)
\end{equation}
where $m_t = m_1$ and $\tau_1(X_i)= X^{(1)}_i-m_1$.

Consider $S_n^{-1}T_n$. Since $T_n$ is bounded, it is straightforward that $T_nR_n'$ satisfies the Condition
$\mathcal{A}$. Then, we multiply every term on the RHS of (\ref{eqn:sninverse2}) except $R'_n$ by the second term on
the RHS of (\ref{eqn:tn}) and identify the resulting terms satisfying the Condition $\mathcal{A}$. Firstly, according
to Lemma \ref{lem:l2}, it is obvious that the second term on the RHS of (\ref{eqn:sninverse2}) times the second term on
the RHS of (\ref{eqn:tn}) satisfies the Condition $\mathcal{A}$. Secondly, the third term on the RHS of
(\ref{eqn:sninverse2}) times the second term on the RHS of (\ref{eqn:tn}) produces
\begin{eqnarray}
   && \frac{2m_s^{-2}}{n}E(\zeta_1(X_1)\tau_1(X_1))
   + \frac{2m_s^{-2}}{n^2}\sum_{i=1}^n(\zeta_1(X_i)\tau_1(X_i)-E(\zeta_1(X_1)\tau_1(X_1)))\cr
   && +\frac{2m_s^{-2}}{n^2}\sum_{1\leq i<j\leq
   n}(\zeta_1(X_i)\tau_1(X_j)+\zeta_1(X_j)\tau_1(X_i)). \nonumber
\end{eqnarray}
In the above expression, the second term satisfies the Condition $\mathcal{A}$ from Lemma \ref{lem:l2}; the third term
is
\begin{eqnarray}
\frac{2m_s^{-2}}{n^2}\sum_{1\leq i<j\leq n}(\zeta_1(X_i)\tau_1(X_j)+\zeta_1(X_j)\tau_1(X_i))
= \frac{2m_s^{-2}}{n(n-1)}\sum_{1\leq i<j\leq n}(\zeta_1(X_i)\tau_1(X_j)+\zeta_1(X_j)\tau_1(X_i))\cr
 -\frac{2m_s^{-2}}{n^2(n-1)}\sum_{1\leq i<j\leq n}(\zeta_1(X_i)\tau_1(X_j)+\zeta_1(X_j)\tau_1(X_i)),
\end{eqnarray}
and the second term on the RHS of the above equation satisfies the Condition $\mathcal{A}$ from Lemma \ref{lem:l2}.
Thirdly, similarly to the treatments in the last term on the RHS of (\ref{eqn:sn2}), we can show that the fourth and
fifth terms on the RHS of (\ref{eqn:sninverse2}) times the second term on the RHS of (\ref{eqn:tn}) both satisfy the
Condition $\mathcal{A}$.

Combining all the associated terms satisfying the Condition $\mathcal{A}$ in $S_n^{-1}T_n$ into a new term $R_n$, we
then prove this lemma.
\end{proof}

\section{Proof of Theorem \ref{thm:rssconvergence}} \label{sec:proofofrssconvergence}


\begin{proof}
We use the notations in Lemma \ref{lem:expansion} and define $\sigma_g=Std(g_1(X_1))$. Moreover, we have
\begin{equation}
   Tr(C_{RSS})=\left(\frac{2}{b(n-1)}\right)\frac{T_n}{S_n}.
\end{equation}

It follows from Lemma \ref{lem:expansion} that $M_n$ is a $U$-statistic of degree $2$ and
\begin{equation}
   Pr\left\{|\sqrt{n}R_n|\geq \frac{\varepsilon}{\sqrt{n}\ln n}\right\}  =  o(1),
\end{equation}
implying that $\sqrt{n}R_n$ converges to $0$ in probability. By $\sigma_g>0$ (due to $\sigma_1>0$) and Theorem $1$ on
page $76$ in \cite{Lee90}, $\frac{1}{2}\sigma_g^{-1}\sqrt{n}M_n$ converges to standard normality as $n\to\infty$. By
letting $W_n=\frac{1}{2}\sigma_g^{-1}\sqrt{n}(M_n+R_n)$, we conclude that as $n\to\infty$, $W_n$ converges in
distribution to standard normality from Theorem $12$ on page $16$ in \cite{Petrov75}.
\end{proof}

\section{Proof of Theorem \ref{thm:rssrate}}\label{sec:proofofrssrate}
\begin{proof}
Let $G_n(x)$ be the distribution function of $\frac{1}{2}\sigma_g^{-1}\sqrt{n}M_n$. From Theorem $1.2$ in \cite{BGZ97},
we have
\begin{equation}
   \sup_x\left|G_n(x)-G(x)\right| = O(n^{-1})
\end{equation}
where
\begin{equation}
   G(x) = \Phi(x)- \left(\frac{\nu_3+\frac{6\sigma_1^4}{m_1}}{6\sigma_g^3}\right)\Phi'''(x)n^{-\frac{1}{2}}.
\end{equation}
Then, we  can obtain
\begin{eqnarray}
 |G_n(x)-\Phi(x)|  &\leq& |G(x)-\Phi(x)| + |G_n(x)-G(x)|\nonumber\\
  &\leq& \left|\left(\frac{\nu_3+\frac{6\sigma_1^4}{m_1}}{6\sigma_g^3 }\right)\Phi'''(x)\right|n^{-\frac{1}{2}} +
  O(n^{-1}).
\end{eqnarray}

For any $\varepsilon>0$, using Lemma $3$ on page $16$ in \cite{Petrov75}, we have
\begin{equation*}
G_n(x-\varepsilon) - Pr\left\{\left|\frac{1}{2}\sigma_g^{-1}\sqrt{n}R_n\right|\geq\varepsilon\right\}  \leq F_n(x) \leq
G_n(x+\varepsilon) + Pr\left\{\left|\frac{1}{2}\sigma_g^{-1}\sqrt{n}R_n\right|\geq\varepsilon\right\}
\end{equation*}
and thus
\begin{eqnarray*}
|F_n(x)- \Phi(x)| \leq \max\left\{|G_n(x\pm\varepsilon)- \Phi(x)|\right\}+
Pr\left\{\left|\frac{1}{2}\sigma_g^{-1}\sqrt{n}R_n\right|\geq\varepsilon\right\}\\
|F_n(x)- \Phi(x)| \leq \max\{|G_n(x\pm\varepsilon)- \Phi(x\pm\varepsilon)|+| \Phi(x\pm\varepsilon)-\Phi(x)|\}+
Pr\left\{\left|\frac{1}{2}\sigma_g^{-1}\sqrt{n}R_n\right|\geq\varepsilon\right\}.
\end{eqnarray*}

By letting $\varepsilon=n^{-\frac{1}{2}}$, we have $Pr\{|\frac{1}{2}\sigma_g^{-1}\sqrt{n}R_n|\geq
n^{-\frac{1}{2}}\}=O(n^{-1})$ from Lemma \ref{lem:expansion}. Moreover, it is easy to show $|\Phi(x\pm
n^{-\frac{1}{2}})-\Phi(x)|=O(n^{-1})$. Then, the theorem is proved.\end{proof}

\section{Proof of Theorem \ref{thm:toaconvergence}} \label{sec:proofoftoaconvergence}
\begin{proof}
We use the same notations as defined in Theorem \ref{thm:rssconvergence}. By letting $X^{(1)}_i=1$ and
$X^{(2)}_i=\alpha_i$, we obtain $\sigma_g=0$, $m_1=1$, $m_2=0.5$ and
\begin{equation}
   Tr(C_{TOA})=\left(\frac{2\sigma^2}{n-1}\right)\frac{T_n}{S_n}.
\end{equation}

The kernel of the $U$-statistic $M_n$ can be expressed as
\begin{equation}
   \varphi(\alpha_1,\alpha_2)= \sum_{i=1}^2 \varphi_i(\alpha_1)\varphi_i(\alpha_2)
\end{equation}
where $\varphi_1(\alpha_1)=\sqrt{2}\cos2\alpha_1$ and $\varphi_2(\alpha_1)=\sqrt{2}\sin2\alpha_1$.

From Theorem $1$ on page $79$ in \cite{Lee90}, we derive that $nM_n$ converges in distribution to $Z_1^2+Z_2^2-2$ where
$Z_1$ and $Z_2$ are independent and standard normal, namely that $nM_n+2$ converges in distribution to a chi-square
random variable of degree $2$. Similarly to the treatments in Theorem \ref{thm:rssconvergence}, we can have
\begin{equation}
   Pr\left\{|nR_n|\geq \frac{\varepsilon}{\ln n}\right\} = o(1),
\end{equation}
implying that $nR_n$ converges to $0$ in probability. By letting $V_n=n(M_n+R_n)+2$, the theorem is immediately proved.
\end{proof}

\ifCLASSOPTIONcaptionsoff
 \newpage
\fi



%
\bibliographystyle{IEEEtran}
\bibliography{IEEEabrv,errorpropagation}

\end{document}